\documentclass[11pt]{article}

\usepackage{fullpage}
\usepackage{amsmath,amsfonts,amsthm,cite}
\usepackage[unicode]{hyperref}
\usepackage{subcaption}
\usepackage{svg}
\usepackage{xspace}
\usepackage{xargs}
\usepackage{graphicx}
\usepackage{overpic}
\usepackage[labelfont=bf,font=small]{caption}
\usepackage[bottom]{footmisc}
\usepackage{tikz}
\usepackage{colortbl}
\usepackage{wrapfig}
\usepackage{enumitem}

\graphicspath{{./svgtiler/},{./figures/}}

\usepackage[colorinlistoftodos]{todonotes}

\newtheorem{theorem}{Theorem}[section]
\newtheorem{lemma}[theorem]{Lemma}

\theoremstyle{definition}
\newtheorem{definition}[theorem]{Definition}

\newcommandx{\Pull}{\textsc{Pull}\xspace}
\newcommandx{\Pullk}[2][1=$k$, 2=?]{\textsc{Pull{#2}}\text{-}#1\xspace}
\newcommandx{\PullkF}[2][1=$k$, 2=?]{\textsc{Pull{#2}}\text{-}#1\text{F}\xspace}
\newcommandx{\PullkFG}[2][1=$k$, 2=?]{\textsc{Pull{#2}}\text{-}#1\text{FG}\xspace}
\newcommandx{\PullkW}[2][1=$k$, 2=?]{\textsc{Pull{#2}}\text{-}#1\text{W}\xspace}
\newcommandx{\PullkWG}[2][1=$k$, 2=?]{\textsc{Pull{#2}}\text{-}#1\text{WG}\xspace}
\newcommandx{\PushPull}{\textsc{PushPull}\xspace}
\newcommandx{\PullPull}{\textsc{PullPull}\xspace}

\newcommand{\NP}{\textsc{NP}\xspace}
\newcommand{\PSPACE}{\textsc{PSPACE}\xspace}

\newcommand{\cIn}[0]{c_\text{in}}
\newcommand{\cOut}[0]{c_\text{out}}
\newcommand{\cIni}[1]{c_{\text{in},#1}}
\newcommand{\cOuti}[1]{c_{\text{out},#1}}
\newcommand{\pIni}[1]{c'_{\text{in},#1}}
\newcommand{\pOuti}[1]{c'_{\text{out},#1}}
\newcommand{\CHX}[0]{\text{CHX}}
\newcommand{\MSC}[0]{\text{MSC}}
\newcommand{\SO}[0]{\text{SO}}

\newcommand{\SD}[0]{\text{SD}}
\newcommand{\SX}[0]{\text{SX}}
\newcommand{\WCX}[0]{\text{WCX}}
\newcommand{\PostSelect}[1]{\textsf{Postselect}(#1)}
\newcommand{\hallway}[0]{\textsf{hallway}}

{\makeatletter \gdef\fps@figure{!htbp}}

\let\realbfseries=\bfseries
\def\bfseries{\realbfseries\boldmath}

\def\emph#1{\textbf{\textit{\boldmath #1}}}

\newif\ifabstract
\abstracttrue
\newif\iffull
\ifabstract \fullfalse \else \fulltrue \fi

\ifabstract

\else

\fi

\newcounter{section-preserve}
\newcounter{theorem-preserve}
\newcommand{\blank}[1]{}
\newtoks\magicAppendix
\magicAppendix={}
\newtoks\magictoks
\newif\iflater
\laterfalse
\long\def\later#1{\iflater#1\else\magictoks={#1}	\edef\magictodo{\noexpand\magicAppendix={\the\magicAppendix \par
			\the\magictoks	}}
	\magictodo\fi}
\long\def\both#1{\iflater#1\else\magictoks={#1}	\edef\magictodo{\noexpand\magicAppendix={\the\magicAppendix \par
			\noexpand\setcounter{theorem-preserve}{\noexpand\arabic{theorem}}			\noexpand\setcounter{theorem}{\arabic{theorem}}			\noexpand\setcounter{section-preserve}{\noexpand\arabic{section}}			\noexpand\setcounter{section}{\arabic{section}}			\noexpand\let\noexpand\oldsection=\noexpand\thesection
			\noexpand\def\noexpand\thesection{\thesection}
			\noexpand\let\noexpand\oldlabel=\noexpand\label
			\noexpand\let\noexpand\label=\noexpand\blank
			\the\magictoks			\noexpand\setcounter{theorem}{\noexpand\arabic{theorem-preserve}}			\noexpand\setcounter{section}{\noexpand\arabic{section-preserve}}			\noexpand\let\noexpand\thesection=\noexpand\oldsection
			\noexpand\let\noexpand\label=\noexpand\oldlabel
	}}
	\magictodo
	\the\magictoks\fi}
\def\magicappendix{\latertrue \the\magicAppendix}

{\makeatletter
	\gdef\xxxmark{		\expandafter\ifx\csname @mpargs\endcsname\relax 		\expandafter\ifx\csname @captype\endcsname\relax 		\marginpar{xxx}		\else
		xxx 		\fi
		\else
		xxx 		\fi}
	\gdef\xxx{\@ifnextchar[\xxx@lab\xxx@nolab}
	\long\gdef\xxx@lab[#1]#2{\textbf{[\xxxmark #2 ---{\sc #1}]}}
	\long\gdef\xxx@nolab#1{\textbf{[\xxxmark #1]}}
	}

\title{Pushing Blocks via Checkable Gadgets: \\ \PSPACE-completeness of Push-1F and Block/Box Dude}
\author{  Hayashi Ani    \thanks{Massachusetts Institute of Technology, Cambridge, MA, USA}
\and
  Lily Chung    \footnotemark[1]
\and
  Erik D. Demaine    \footnotemark[1]
\and
  Jenny Diomidova    \footnotemark[1]
\and
  Della Hendrickson    \footnotemark[1]
\and
  Jayson Lynch    \footnotemark[1]
}
\date{}

\begin{document}

\maketitle

\begin{abstract}
We prove \PSPACE-completeness of the well-studied pushing-block puzzle
Push-1F, a theoretical abstraction of many video games (introduced in 1999).
The proof also extends to Push-$k$ for any $k \ge 2$.
We also prove \PSPACE-completeness of two versions of the recently studied
block-moving puzzle game with gravity, Block Dude --- a video game dating back to 1994 ---
featuring either liftable blocks or pushable blocks.
Two of our reductions are built on a new framework for ``checkable'' gadgets,
extending the motion-planning-through-gadgets framework to support gadgets
that can be misused, provided those misuses can be detected later.
\end{abstract}

\section{Introduction}

In the \emph{Push} family of pushing-block puzzles,
introduced by O'Rourke and The Smith Problem Solving Group
in 1999 \cite{O'Rourke99},
a $1 \times 1$ agent must traverse a unit-square grid,
some cells of which have a ``block'',
from a given start location to a given target location.
Refer to Figure~\ref{fig:Push-1F}.
In \emph{Push-$k$} \cite{Push100,demaine2003pushing},
the agent's move (horizontal or vertical by one square)
can \emph{push} up to $k$ consecutive blocks by one square,
provided that there is an empty square on the other side.
In \emph{Push-$*$}
(described in \cite{Dhagat-O'Rourke-1992,O'Rourke99}
and analyzed in
\cite{bremner1994motion,hoffmann-2000-pushstar,demaine2003pushing}),
there is no limit to the number of consecutive blocks that can be pushed.
In the \emph{-F} variation (described in
\cite{Dhagat-O'Rourke-1992,O'Rourke99,Push100} but first
given notation in \cite{demainepush2F}), some of the blocks are \emph{fixed}
in the grid, meaning they cannot be traversed or pushed by the agent
or other blocks.
Push-1F has the same allowed moves as the famous \emph{Sokoban}
puzzle video game, invented in 1982 and analyzed at FUN 1998 \cite{sokoban},
but crucially Push-1F's goal is for the agent to reach a target location,
which is much simpler than Sokoban's ``storage'' goal where the blocks must be pushed to certain locations.

\begin{figure}
  \centering
  \subcaptionbox{}{\includegraphics[width=0.16\linewidth]{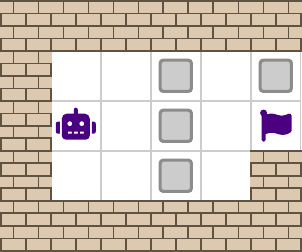}}\hfill
  \subcaptionbox{}{\includegraphics[width=0.16\linewidth]{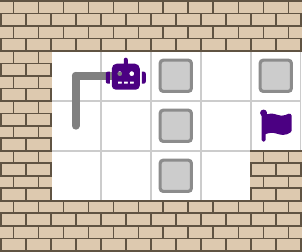}}\hfill
  \subcaptionbox{}{\includegraphics[width=0.16\linewidth]{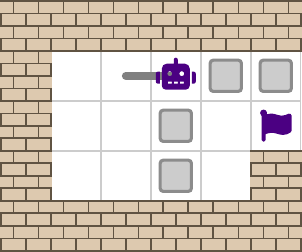}}\hfill
  \subcaptionbox{}{\includegraphics[width=0.16\linewidth]{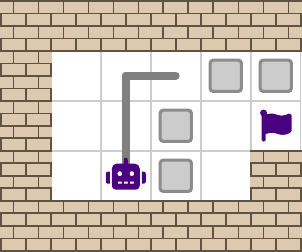}}\hfill
  \subcaptionbox{}{\includegraphics[width=0.16\linewidth]{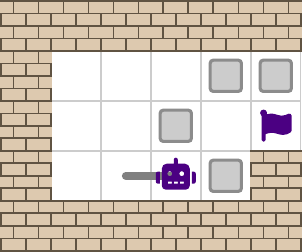}}\hfill
  \subcaptionbox{}{\includegraphics[width=0.16\linewidth]{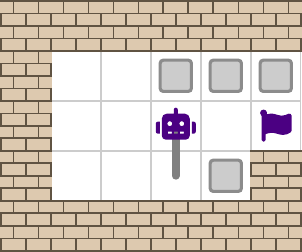}}
  \caption{Sample Push-1F puzzle and solution sequence.
    In steps (c) and (e), for example, the\ agent cannot push right again.
    The agent is drawn as a robot head;
    the traversed path between steps is drawn as a gray line;
    pushable blocks are drawn as boxes;
    fixed blocks are drawn as brick walls; and
    the goal location is drawn as a flag.
    \textsl{Robot and flag icons from Font Awesome under CC BY 4.0 License.}
  }
  \label{fig:Push-1F}
\end{figure}

In this paper, we prove that Push-1F is \PSPACE-complete,
settling an open problem from \cite{Push100,demainepush2F},
and complementing previous results showing \PSPACE-hardness for Push-$k$F for $k \geq 2$
over 20 years ago \cite{demainepush2F}
and \PSPACE-hardness for Push-$*$F 30 years ago
\cite{bremner1994motion}.

To gain some intuition about why Push-1F is so difficult to prove
\PSPACE-hard, and how we surmount that difficulty, consider the
attempt at a ``diode'' gadget in Figure~\ref{fig:broken diode}.
The goal of this gadget is to allow repeated traversals from the
left entrance to the right (as in Figure~\ref{fig:broken diode:forward}),
while always preventing ``backward'' traversal from the right to the left
(as in Figure~\ref{fig:broken diode:backward}).
But given the opportunity for forward traversal, the agent can instead
``break'' the gadget to allow future forward and backward traversal
(as in Figure~\ref{fig:broken diode:break}).

\begin{figure}
  \centering
  \subcaptionbox{Gadget}
    {\includegraphics[width=0.18\linewidth]{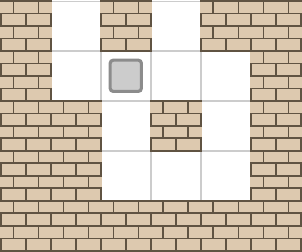}}\hfill
  \subcaptionbox{\label{fig:broken diode:forward}Intended forward traversal}{    \includegraphics[width=0.18\linewidth]{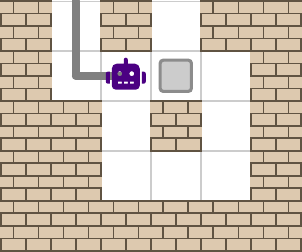}
    \includegraphics[width=0.18\linewidth]{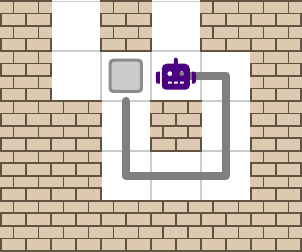}  }\hfill
  \subcaptionbox{\label{fig:broken diode:backward}Backward traversal impossible}
    {\includegraphics[width=0.18\linewidth]{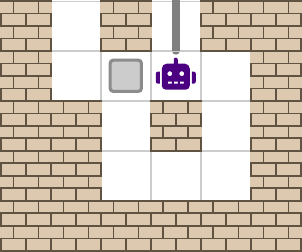}}\hfill
  \subcaptionbox{\label{fig:broken diode:break}Breaking the gadget}
    {\includegraphics[width=0.18\linewidth]{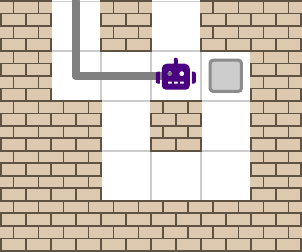}}
  \caption{A broken Push-1F diode gadget.}
  \label{fig:broken diode}
\end{figure}

To solve this problem, we introduce the idea of a \emph{checkable gadget}
where, after the agent completes the ``main'' gadget traversal puzzle,
the agent is forced (in order to solve the overall puzzle)
to do a specified sequence of \emph{checking} traversals of every gadget,
all of which must succeed in order to solve the overall puzzle.
If designed well, these checking traversals can detect whether a
gadget was previously ``broken'', and allow traversal only if not.
In the case of Figure~\ref{fig:broken diode}, one can think of the gadget as a four-location gadget (the top three rows) which has its bottom two locations connected.
This four-location gadget is ``checkable'':
we will demand that,
after completing the main puzzle,
the agent follows the two checking traversals shown in Figure~\ref{fig:checkable diode}.
In order for these checking traversals to both be possible,
the agent cannot push the block into either corner, preventing the agent
from breaking the gadget during the main gadget traversal puzzle.
We call this process of removing broken states from a gadget by demanding that the checking traversals remain legal \emph{postselection}.\footnote{In quantum computing, for example, ``postselection is the power of
  discarding all runs of a computation in which a given event does not occur''
  \cite{Aaronson-postselection}.  In probability theory, postselection is
  equivalent to conditioning on a particular event.}

\begin{figure}
  \centering
  \scriptsize    \vspace*{4ex}

  \subcaptionbox{Checkable gadget}{    \begin{overpic}[width=0.18\linewidth]{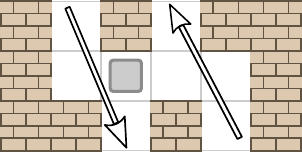}
      \put(25,65){\makebox(0,0){\strut check}}
      \put(25,55){\makebox(0,0){\strut 1 in}}
      \put(42.5,-6){\makebox(0,0){\strut check}}
      \put(42.5,-16){\makebox(0,0){\strut 1 out}}
      \put(57.5,65){\makebox(0,0){\strut check}}
      \put(57.5,55){\makebox(0,0){\strut 2 out}}
      \put(75,-6){\makebox(0,0){\strut check}}
      \put(75,-16){\makebox(0,0){\strut 2 in}}
    \end{overpic}
    \vspace*{4ex}
  }\hfill
  \subcaptionbox{Successful checks}{    \includegraphics[width=0.18\linewidth]{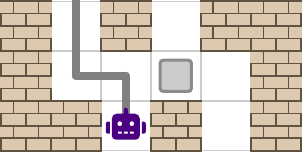}
    \includegraphics[width=0.18\linewidth]{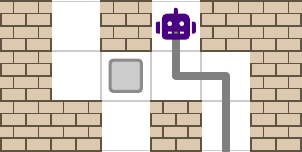}
    \vspace*{4ex}
  }\hfill
  \subcaptionbox{Failed checks}{    \includegraphics[width=0.18\linewidth]{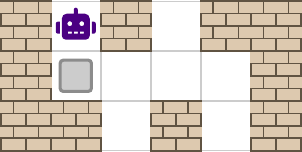}
    \includegraphics[width=0.18\linewidth]{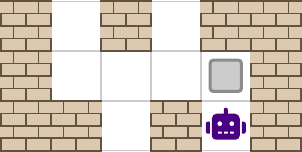}
    \vspace*{4ex}
  }
  \caption{The top three rows
    of the Push-1F diode gadget of Figure~\ref{fig:broken diode}, as a checkable gadget.
    The checking traversals are
    ``check 1 in $\rightarrow$ check 1 out'' and
    ``check 2 in $\rightarrow$ check 2 out'', denoted by the hollow arrows.}
  \label{fig:checkable diode}
\end{figure}
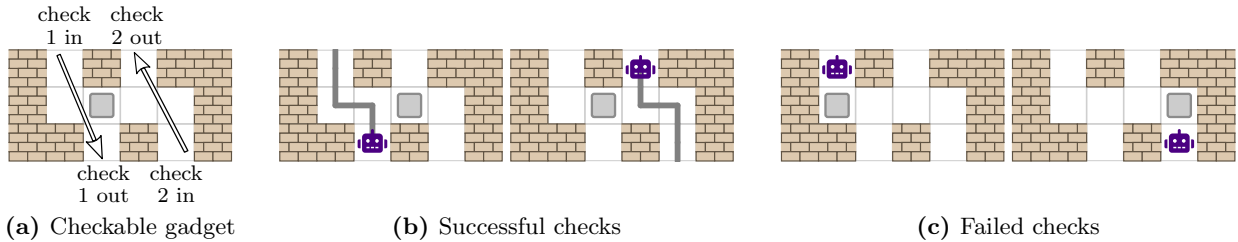

We develop a general framework of checkable gadgets that enable a reduction
to focus on the main gadget traversal puzzle, assuming all gadgets remain
unbroken (i.e., the checking traversals remain possible at the end),
while the framework ensures that the agent makes these checking traversals
at the end (without other unintended traversals).
This framework builds upon the motion-planning-through-gadgets framework
introduced at FUN 2018 \cite{demaine2018computational} and developed
further in \cite{demaine2018general,ani2020pspace,doors,lynch2020framework,hendrickson2021gadgets}
to handle checkable gadgets.

We show that the reduction to Push-1F can be applied to Push-$k$
(without fixed blocks) for any $k \ge 2$, by a simple change of gadgets.
This leaves Push-1 and Push-$*$ as the only open problems
in terms of whether we need fixed blocks for PSPACE-completeness.

\label{sec:Block Pushing Models}
We also apply the motion-planning-through-gadgets framework
to resolve the complexity of \emph{Block Dude},
a puzzle video game made over a dozen times on many platforms,
originally under the name ``Block-Man 1'' (Soleau Software, 1994); see
\cite{barr2021block} for details.
Barr, Chung, and Williams \cite{barr2021block} recently formalized this
game's mechanics, along with several variations,
and proved them all NP-hard.
In this paper, we prove PSPACE-completeness of three of these variations,
including the original video game mechanics:
\begin{enumerate}
\item \emph{BoxDude} is like Push-1 but where
  all pushable blocks and the agent experience gravity,
  falling straight down whenever they have blank spaces below them.
  In addition to moving horizontally left or right,
  the agent can ``climb'' on top of horizontally adjacent blocks
  (be they pushable or fixed), provided the square above the agent is empty.
  See Figure~\ref{fig:BoxDude}.

\begin{figure}
  \centering
  \subcaptionbox{Pushing}{    \includegraphics[width=0.15\linewidth]{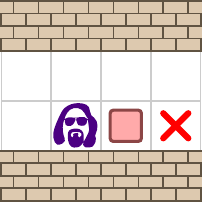}
    \includegraphics[width=0.15\linewidth]{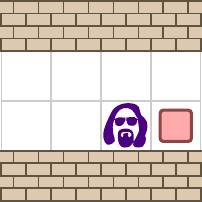}
  }\hfil
  \subcaptionbox{Climbing}{    \includegraphics[width=0.15\linewidth]{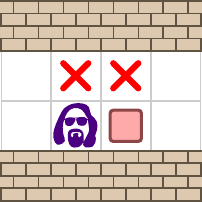}
    \includegraphics[width=0.15\linewidth]{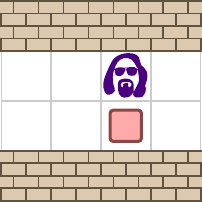}
  }
  \caption{Mechanics for BoxDude, with pushable boxes shown in red.
    Squares marked with a red $\times$ must be empty
    for the move to be possible.}
  \label{fig:BoxDude}
\end{figure}

\item In \emph{BlockDude} (as in the Block Dude video games),
  blocks cannot be pushed;
  instead, nonfixed blocks can be ``picked up'' by the agent from a
  horizontally adjacent position to the position immediately above the agent,
  provided that that position and the intermediate diagonal position are empty.
  See Figure~\ref{fig:BlockDude}.
  The agent can then carry one such block to another location
  (provided the ceiling offers height-2 clearance),
  and then drop the block in front of them, again
  provided that that position and the intermediate diagonal position are empty.    \footnote{A complication in some implementations of the game is that the
    agent can only pick up or drop the block in front of them, with the
    agent's orientation determined by their previous move.
    (Some implementations allow turning around in place.)
    This detail will not affect our results.}
    They can also stack the block on top of another block.
  If the agent tries to move past a low ceiling while carrying a block, the block will be dropped behind them.

\begin{figure}
  \centering
  \subcaptionbox{Lifting a block}{    \includegraphics[width=0.15\linewidth]{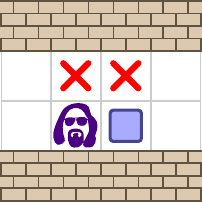}
    \includegraphics[width=0.15\linewidth]{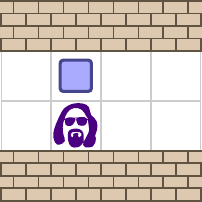}
  }\hfill
  \subcaptionbox{Carrying a lifted block}{    \includegraphics[width=0.15\linewidth]{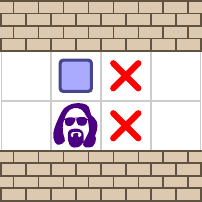}
    \includegraphics[width=0.15\linewidth]{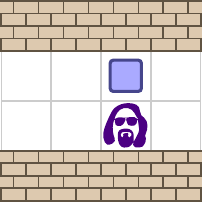}
  }\hfill
  \subcaptionbox{Low clearance when carrying}{    \includegraphics[width=0.15\linewidth]{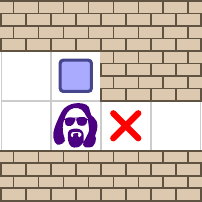}
    \includegraphics[width=0.15\linewidth]{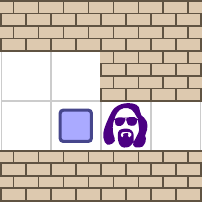}
  }

  \subcaptionbox{Climbing with a block}{    \includegraphics[width=0.15\linewidth]{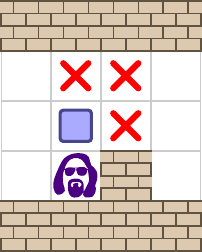}
    \includegraphics[width=0.15\linewidth]{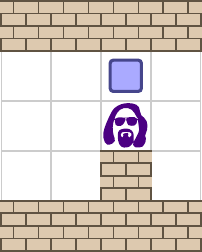}
  }\hfill
  \subcaptionbox{Dropping a lifted block}{    \includegraphics[width=0.15\linewidth]{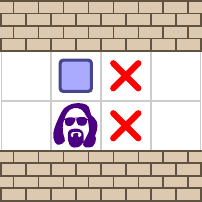}
    \includegraphics[width=0.15\linewidth]{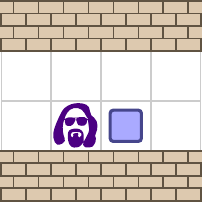}
  }\hfill
  \subcaptionbox{Stacking a lifted block}{    \includegraphics[width=0.15\linewidth]{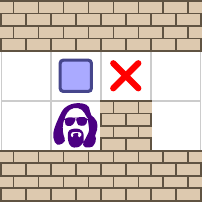}
    \includegraphics[width=0.15\linewidth]{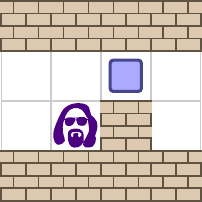}
  }
  \caption{Mechanics for BlockDude, with liftable blocks shown in blue.
    Squares marked with a red $\times$ must be empty
    for the following move to be possible.}
  \label{fig:BlockDude}
\end{figure}

\item In \emph{BloxDude}, nonfixed blocks can be pushed (as in BoxDude)
  and/or picked up (as in BlockDude).
\end{enumerate}
These proofs are simpler in that they do not require our checkable gadgets
framework.  Therefore we present them first, along with a proof that
Push-$*$F is PSPACE-hard, adapted from a 30-year-old unpublished
manuscript \cite{bremner1994motion} which we include here for completeness.

The other variations described in \cite{barr2021block}, called $\cdots$Duderino
instead of $\cdots$Dude, change the goal of a puzzle to place the $k$ nonfixed
blocks into $k$ specified storage locations, as in Sokoban.
We leave open the complexity of BoxDuderino, BlockDuderino, and BloxDuderino.

All of the games we consider can easily be simulated in polynomial space,
and thus are in NPSPACE $=$ PSPACE by Savitch's Theorem.
Proving PSPACE-hardness is much more complicated, and is the goal of this paper.

The rest of this paper is organized as follows.
In Section~\ref{sec:gadget model}, we review the
motion-planning-through-gadgets framework.
In Section~\ref{sec:Push-*F}, we prove that Push-$*$F is \PSPACE-complete
using a reduction from motion-planning-through-gadgets
using gadgets from an old unpublished proof \cite{bremner1994motion}.
In Section~\ref{sec:BlockDude}, we prove that BlockDude and BloxDude
are \PSPACE-complete
using standard reductions from motion-planning-through-gadgets.
In Section~\ref{sec:checkable gadget framework}, we develop our checkable gadget
framework.
In Section~\ref{sec:BoxDude}, we prove that BoxDude is \PSPACE-complete
using our checkable gadget framework. 
In Section~\ref{sec:Push-1F}, we prove that Push-1F is \PSPACE-complete
via a much more involved application of our checkable gadget framework.
In Section~\ref{sec:Push-k}, we prove that Push-$k$ is \PSPACE-complete
for any $k \ge 2$, similar to Section~\ref{sec:Push-1F} but replacing the gadgets.
We conclude with open problems in Section~\ref{sec:Open Problems}.

\section{Gadgets Framework}
\label{sec:gadget model}

The \emph{motion-planning-through-gadgets framework} is an abstract motion planning model used for proving computational hardness results.
Here we give the definitions and results we need for this paper;
see \cite{demaine2018general,lynch2020framework,hendrickson2021gadgets}
for more details.

A \emph{gadget} $G$ consists of
a finite set $Q(G)$ of \emph{states},
a finite set $L(G)$ of \emph{locations} (entrances/exits),
and a finite set $T(G)$ of \emph{transitions} of the form $(q,a) \to (r,b)$
where $q,r \in Q(G)$ are states and $a,b \in L(G)$ are locations.
The transition $(q,a) \to (r,b) \in T(G)$ means that an agent can \emph{traverse}
the gadget when it is in state $q$ by entering at location $a$ and exiting
at location~$b$ which changes the state of the gadget from $q$ to~$r$.
We use the notation $a \to b$ for a traversal by the agent that does
not specify the state of the gadget before or after the traversal.
A \emph{traversal sequence} $[a_1\to b_1,\dots,a_k\to b_k]$
on the locations $L(G)$ is \emph{legal} from state $s_0$
if there is a corresponding sequence of transitions
$[(a_1,s_0) \to (b_1,s_1), \dots, (a_k,s_{k-1}) \to (b_k,s_k)]$,
where each start state of each transition matches the end state
of the previous transition ($s_0$ for the first transition).
We define gadgets in figures using a \emph{state diagram} which gives,
for each state $q \in Q$,
a labeled directed multigraph $G_q = (L(G),E_q)$ on the locations,
where a directed edge $(a,b)$ with label $r$ represents the transition
$(q,a) \to (r,b) \in T(G)$.

Figure~\ref{fig:L2T} shows the state diagram of a key gadget
called the \emph{locking 2-toggle} \cite{demaine2018general}.
This gadget has four locations (drawn as dots) and three states $1,2,3$.
The central state, $2$, allows for two different transitions.
Each of those transitions takes the gadget to a different state,
from which the only transition returns the agent to the prior location
and returns the gadget to state~$3$.

\begin{figure}
		\centering	\includegraphics[scale=1]{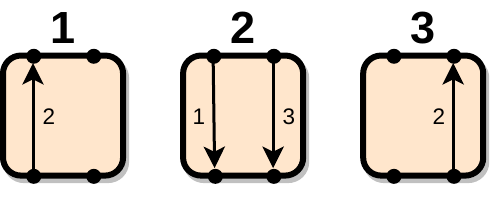}
	\caption{State diagram for the locking 2-toggle gadget. Each box represents the gadget in a different state, in this case labeled with the numbers $1,2,3$. Dots represent the four locations of the gadget. Arrows represent transitions in the gadget and are labeled with the states to which those transitions take the gadget. In state 2, the agent can traverse either tunnel going down, which blocks off both downward traversals until the agent reverses that traversal.}
	\label{fig:L2T}
	\end{figure}

A \emph{system of gadgets} $S$ consists of a set of gadgets,
an initial state for each gadget,
and a \emph{connection graph} on the gadgets' locations.
If two locations $a,b$ of two gadgets (possibly the same gadget) are connected
by a path in the connection graph, then an agent can traverse freely between
$a$ and~$b$ (outside the gadgets).\footnote{Equivalently, we can think of identifying locations $a$ and $b$
  topologically, thereby contracting the connected components of the
  connection graph.
  Alternatively, if we think of the gadgets as individual ``levels'',
  then the connection graph is like an ``overworld'' map
  connecting the levels together.}
We call edges of the connection graph \emph{hallways},
and for clarity in figures,
we add extra vertices to the connection graph called \emph{branching hallways},
which we can equivalently think of as a one-state gadget that has transitions
between all pairs of locations.
A \emph{system traversal} is a sequence of traversals
$a_1\to b_1,\dots,a_k\to b_k$, each on a potentially different gadget in~$S$,
where the connection graph has a path from $b_i$ to $a_{i+1}$ for each~$i$.
We write such a traversal as $a_1 \to^* b_k$,
ignoring the intermediate locations.
A system traversal is \emph{legal} if the restriction to traversals on
a single gadget $G$ is a legal traversal sequence from the initial state
of $G$ assigned by $S$, for every $G$ in~$S$.
Note that gadgets are ``local'' in the sense that traversing a gadget does
not change the state (and thus traversability) of any other gadgets.

The \emph{reachability} or \emph{1-player motion planning} problem
with a finite set of gadgets $\mathcal G$
asks whether there is a legal system traversal $s \to^* t$
from a given start location $s$ to a given goal location $t$
(by a single agent) in a given system of gadgets~$S$,
which contains only gadgets from $\mathcal G$.

Because we are working with 2D games, we also consider
\emph{planar motion planning}, where every gadget additionally has a specified
cyclic ordering of its vertices
and the system of gadgets is embedded in the plane without intersections.
More precisely, a system of gadgets is \emph{planar} if the following
construction produces a planar graph:
(1)~replace each gadget with a wheel graph,
which has a cycle of vertices corresponding to the locations on the gadget
in the appropriate order, and a central vertex connected to each location;
and (2)~connect locations on these wheels with edges
according to the connection graph.
In \emph{planar reachability},
we restrict to planar systems of gadgets.
Note that this definition allows rotations and reflections of gadgets,
but no other permutation of their locations.

\subsection{Simulation}

To define a notion of gadget simulation, we can think of a system of gadgets
as being characterized by its set of possible traversal sequences
(as formalized by the related \emph{gizmo} framework
of \cite{hendrickson2021gadgets}).

\begin{definition} \label{def:local simulation}
  A \emph{(local) simulation} of a gadget $G$ in state $q$ consists of
  a system $S$ of gadgets, together with an injective function $m$ mapping
  every location of $G$ to a distinct location in~$S$, such that
  a traversal sequence $[a_1\to b_1,\dots, a_k\to b_k]$
  on the locations in $G$ is legal from state $q$
  if and only if
  there exists a sequence of system traversals
  $m(a_1) \to^* m(b_1), \dots, m(a_k) \to^* m(b_k)$
  that is legal in the sense that
  the concatenation of the restrictions
  of the system traversals $m(a_i) \to^* m(b_i)$
  to traversals on a single gadget $G$
  is a legal traversal sequence for $G$ from the initial state of $G$
  assigned by~$S$, for every $G$ in~$S$.

  A \emph{planar simulation} of a gadget $G$ in state~$q$
  is a simulation $(S,m)$ where
  $S$ is furthermore a planar system of gadgets,
  and the cyclic order of locations of $G$ must map via $m$ to
  locations in cyclic order around the outside face of~$S$.

  A [planar] simulation of an entire gadget $G$ consists of
  a [planar] simulation of $G$ in state $q$, for all states $q \in Q(G)$,
  that differ only in their assignments of initial states.
  A finite set $\mathcal G$ of gadgets \emph{[planarly] simulates}
  a gadget $G$ if there is a [planar] simulation of $G$
  using only gadgets in $\mathcal G$.
\end{definition}

These definitions of simulation imply that,
if we take a larger system of gadgets
and replace each instance of gadget $G$ with the system~$S$
using the appropriate initial states
(matching up locations that correspond via~$m$),
then the entire system behaves equivalently.
In particular, this substitution preserves reachability of locations from one another.
Furthermore, if the larger system and the simulation are both planar,
then the full resulting system is planar. More formally:

\begin{lemma}\label{lem:simulation reduction}
Let $H$ be a gadget,
and let $\mathcal G$ and $\mathcal G'$ be finite sets of gadgets.
If $\mathcal G$ [planarly] simulates $H$,
then there is a polynomial-time reduction\footnote{Throughout this paper,
  reductions are \emph{many-one}/\emph{Karp}:
  a reduction from $A$ to $B$ maps an instance of $A$ to an equivalent
  (in terms of decision outcome) instance of $B$.}
from [planar] reachability with $\{H\} \cup \mathcal G'$
to [planar] reachability with $\mathcal G\cup\mathcal G'$.
\end{lemma}

\subsection{Known Hardness Results}

We can now formally state the problems we will reduce from in this paper.

In Section~\ref{sec:BlockDude}, we use the locking 2-toggle to show PSPACE-completeness of BlockDude puzzles.
\begin{theorem} \label{thm:L2T_PSPACE}
  {\rm \cite[Theorem 10]{demaine2018general}}
	Planar reachability with any interacting-$k$-tunnel reversible deterministic gadget is PSPACE-complete.
\end{theorem}
The locking 2-toggle is an example of an interacting-$k$-tunnel reversible deterministic gadget \cite{demaine2018general} and thus we obtain PSPACE-completeness of planar reachability with the locking 2-toggle. We recommend readers interested in this more general dichotomy to refer to \cite{demaine2018general}.

\begin{figure}
  \centering
  \includegraphics{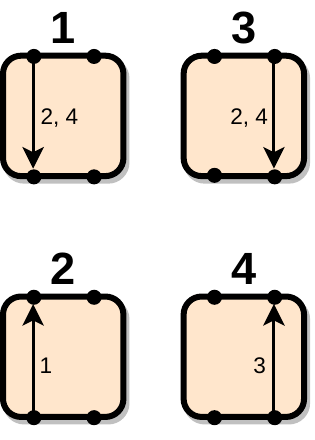}
  \caption{State diagram for a nondeterministic locking 2-toggle.  From state 1, the left tunnel can be traversed so as to leave the gadget in either state 2 or state 4.  Formally, in the multigraph for state 1 there are two different edges, one labeled 2 and the other labeled 4.}
  \label{fig: nL2T}
\end{figure}

We also use the nondeterministic locking 2-toggle shown in Figure~\ref{fig: nL2T}. This is used in Section~\ref{sec:BoxDude} to show PSPACE-completeness of BoxDude puzzles. Its behavior resembles that of the locking 2-toggle, but because it is not deterministic it is not covered by the prior theorem.
\begin{theorem} \label{thm:nL2T_PSPACE}
  {\rm \cite[Theorem 3.1]{ani2020pspace}}
  Planar reachability with the nondeterministic locking 2-toggle is PSPACE-complete.
\end{theorem}

The final main gadget we will make use of is a type of self-closing door shown in Figure~\ref{fig:self-closing door}. This gadget will be used in our result on Push-1F in Section~\ref{sec:Push-1F}.
\begin{theorem}\label{thm:dir-self-closing-door-pspace}
  {\rm \cite[Theorem 4.2]{doors}}
	Planar reachability with any normal or symmetric self-closing door
	is PSPACE-hard.
\end{theorem}

\begin{figure}
	\centering
	\includegraphics{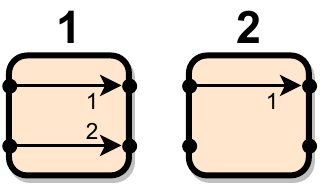}
	\caption{State diagram for the directed open-optional self-closing door.  The door must be opened by visiting its opening location before every traversal.}
	\label{fig:self-closing door}  
\end{figure}

\section{Push-$*$F is PSPACE-complete}
\label{sec:Push-*F}

In this section, we show that Push-$*$F is PSPACE-complete
using a reduction from planar reachability with self-closing doors
(the gadget shown in Figure~\ref{fig:self-closing door}),
which is PSPACE-complete by Theorem~\ref{thm:dir-self-closing-door-pspace}.
Recall that this model allows for the agent to
push any number of movable blocks in a row.
The gadgets are based on gadgets from \cite{bremner1994motion},
which reduced from TQBF but happened to build a door gadget along the way.
We include this proof because the only known proof that Push-$*$F is PSPACE-complete
\cite{bremner1994motion} remains unpublished.\footnote{\cite{demainepush2F} (which proves Push-$k$F PSPACE-complete for all $k \geq 2$)
  claims to also prove PSPACE-completeness of Push-$*$F,
  but their gadgets in fact seem to depend on the specific (finite) value
  of~$k$. Specifically, consider their Figure 2c, designed for Push-3F;
  in Push-4F, $I \to O$ and $LI \to LO$ can be traversed
  without visiting $U$ in between, so the gadget
  requires changes for increased agent strength.}

Figure~\ref{fig:Push-*F} shows the gadgets.
First, Figure~\ref{fig:Push-*F diode} gives a \emph{diode} gadget,
which can be traversed repeatedly from left to right but not vice versa.
Then, we build a self-closing door out of these diode gadgets,
represented by the arrows in Figures~\ref{fig:Push-*F scd open}
(the open state) and~\ref{fig:Push-*F scd closed} (the closed state).
In the open state,
the agent can traverse from the middle-left to the top-left entrance as shown in Figure~\ref{fig:Push-*F close steps},
putting the door into the closed state.
Attempting to perform this traversal in the closed state fails when trying to push
the block left (the third step in Figure~\ref{fig:Push-*F close steps}).
In the closed state, the agent can open the door from the bottom-left entrance,
as shown in Figure~\ref{fig:Push-*F open steps}.

\begin{figure}
  \centering
  \subcaptionbox{\label{fig:Push-*F diode}Diode, represented by arrows in (b) and (c).
    Based on \cite[Figure~4]{bremner1994motion}.}
  [1.75in]
  {\includegraphics[scale=0.6]{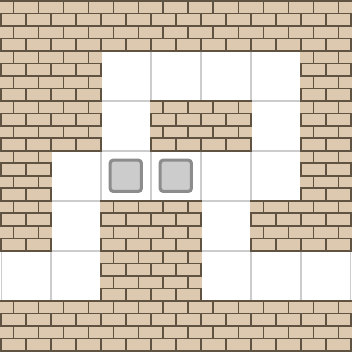}}
  \hfill
  \subcaptionbox{\label{fig:Push-*F scd open}Self-closing door in open state.
    Based on \cite[Figure~5]{bremner1994motion}.}
  {\includegraphics[scale=0.6]{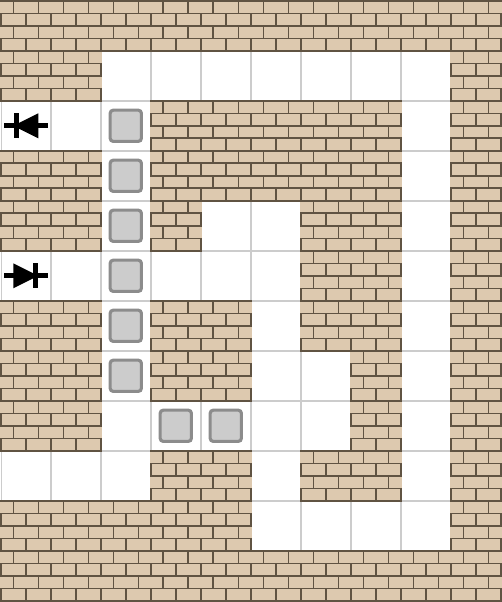}}
  \hfill
  \subcaptionbox{\label{fig:Push-*F scd closed}Self-closing door in closed state.
    Based on \cite[Figure~5]{bremner1994motion}.}
  {\includegraphics[scale=0.6]{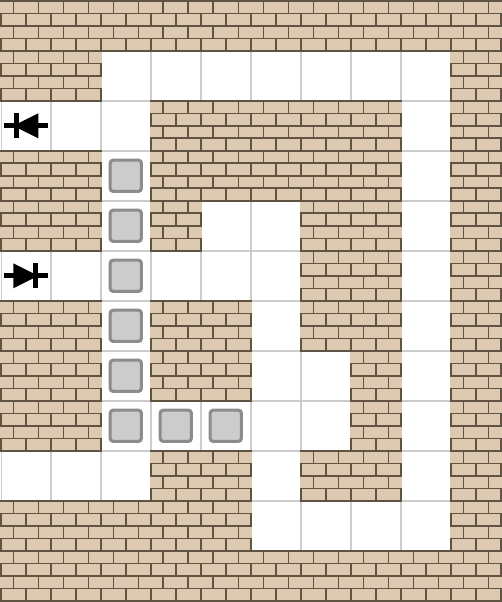}}

  \medskip
  \subcaptionbox{\label{fig:Push-*F close steps}Self-closing traversal sequence.}
  {\includegraphics[scale=0.35]{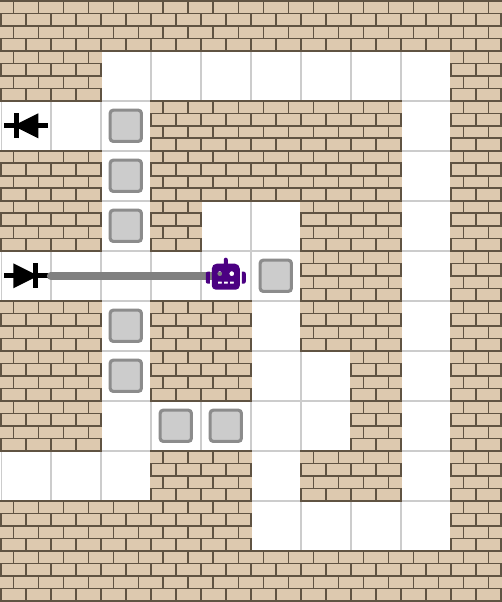}\quad
  \includegraphics[scale=0.35]{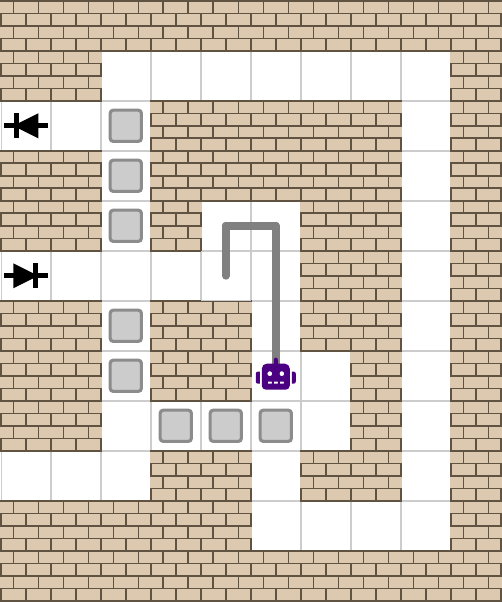}\quad
  \includegraphics[scale=0.35]{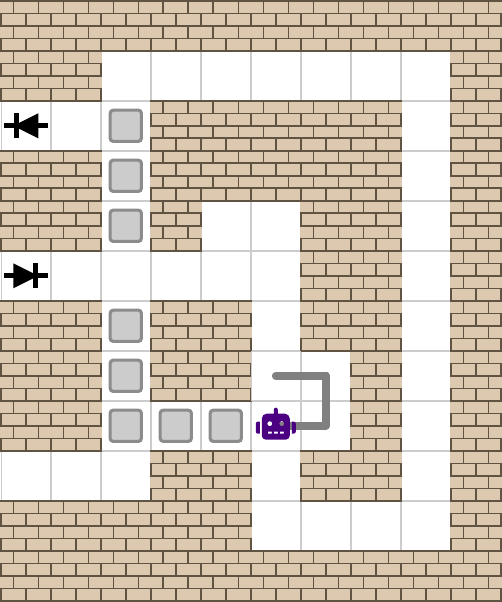}\quad
  \includegraphics[scale=0.35]{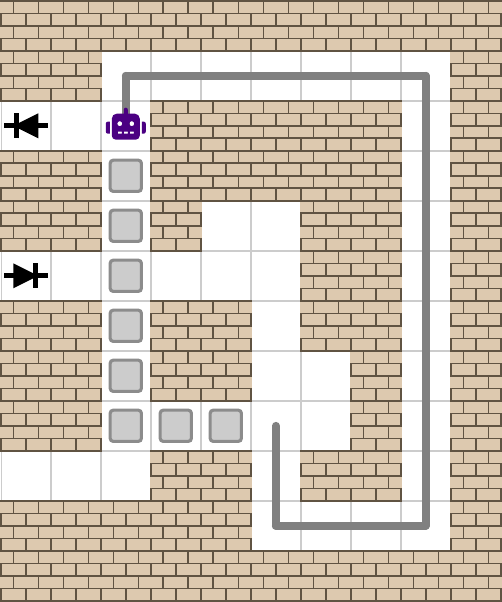}}
  \hfill
  \subcaptionbox{\label{fig:Push-*F open steps}Open sequence.}
  {\includegraphics[scale=0.35]{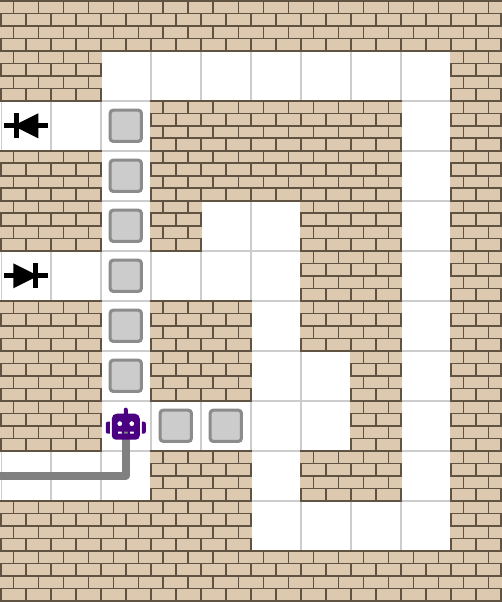}}
  \caption{Gadgets proving PSPACE-completeness of Push-$*$F.}
  \label{fig:Push-*F}
\end{figure}

\section{BlockDude and BloxDude are PSPACE-complete}
\label{sec:BlockDude}

In this section, we show that BlockDude and BloxDude are PSPACE-complete
using a reduction from planar reachability with locking 2-toggles
(the gadget shown in Figure~\ref{fig:L2T}), which is PSPACE-complete by
Theorem~\ref{thm:L2T_PSPACE}. Recall from Section~\ref{sec:Block Pushing Models} in this model blocks can be picked up by BlockDude from an adjacent square. BloxDude allows both picking up and pushing blox, and the reduction will be a small modification to the BlockDude proof. 

We will build hallways allowing the player to move between connected locations on gadgets. To connect more than two locations, we need a branching hallway, which is shown in Figure~\ref{fig: block branching}. This allows the player to freely move between any of the three entrances.

\begin{figure}
	\centering
	\includegraphics{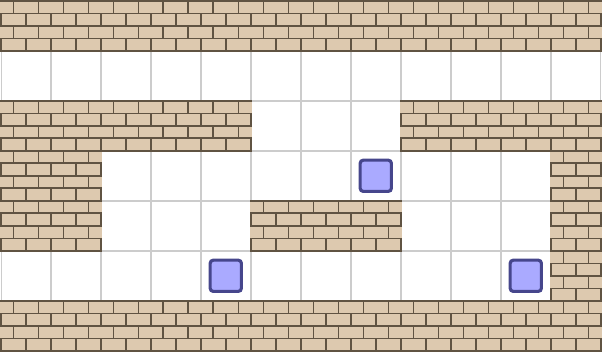}
	\caption{A branching hallway for BlockDude. Blue squares represent blocks (which can be picked up).}
	\label{fig: block branching}
\end{figure}

We now describe how the player can use the branching hallway in a way that always lets them move between any of its entrances. Whenever the player is outside the branching hallway, both bottom blocks will be in their original positions, and the top block will be somewhere on the middle platform, depending on the most recently taken exit.
When the player arrives at the branching hallway, they will first move the top block to the right side of the middle platform (the position in Figure~\ref{fig: block branching}). The only case where this is nontrivial is when the player enters at the bottom with the top block on the left. In this case, the player can go under the middle platform and climb up from the right by moving both bottom blocks. Then they can pick up the top block and step back down on the right, causing the carried block to fall onto the right end of the middle platform. Finally, they can reset the bottom blocks and return to the bottom entrance. 
Once the top block is on the right, the player can take whichever exit they need. If they take the top left exit, they will move the top block to the left first.

To embed an arbitrary planar graph in BlockDude, we also need to be able to
turn hallways and in particular to make vertical hallways despite gravity.
Fortunately, the branching hallway in Figure~\ref{fig: block branching}
can achieve both goals. If we ignore the top-right entrance, the agent
can turn around and make some vertical progress.  By chaining these switchbacks
in alternating orientation, we can build an arbitrarily tall vertical hallway.

To complete the proof of PSPACE-hardness, we only need to build a locking 2-toggle. 
We will construct the locking 2-toggle out of simpler pieces, as shown in Figure~\ref{fig: block l2t compressed}. The simpler pieces are two kinds of 1-toggle: one just for the player, and one that the player can carry a block through. The state diagram for a 1-toggle is given in Figure~\ref{fig: 1-toggle_state}. When the player arrives at (say) the bottom left entrance, they can grab the block in the middle and bring it to the left side, and use it to reach the top left entrance. With the block stuck on the left, the right side cannot be traversed until the player returns to the top left, puts the block back, and exits the bottom left. The player cannot move through this gadget in any way not allowed by a locking 2-toggle. They may leave the block on the left side when the exit the bottom left, but this does not achieve anything; it only prevents them from traversing the right side.

\begin{figure}
	\centering
	\includegraphics{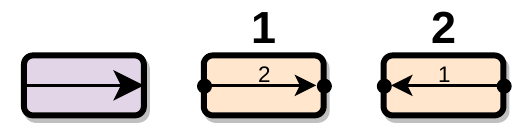}
	\caption{Icon and state diagram for the 1-toggle.  Leftwards and rightwards traversals must alternate.}
	\label{fig: 1-toggle_state}
\end{figure}

\begin{figure}
	\centering
	\includegraphics{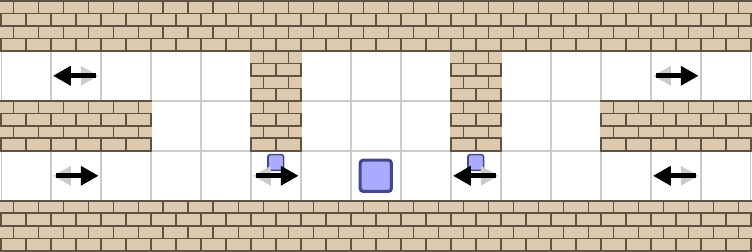}
	\caption{The schematic for our locking 2-toggle for BlockDude. Arrows with a faded backward arrowhead are 1-toggles. Only the player can go through the 1-toggle unless it has a block icon above the arrow, in which case the player can carry a block through.}
	\label{fig: block l2t compressed}
\end{figure}

Our 1-toggle for just the player is shown in Figure~\ref{fig: block 1-toggle}. In the state shown, the player can not enter on the right. If they enter on the left, they can move the blocks to exit on the right, but in doing so must block the left entrance. Because of the 1-high hallways, the player can not bring a block through this gadget.

\begin{figure}
	\centering
	\includegraphics{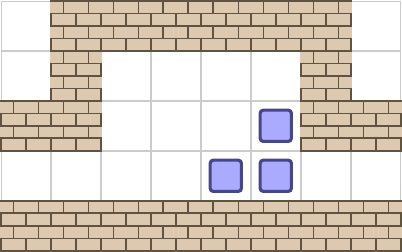}
	\caption{A 1-toggle for BlockDude, currently traversable from left to right.}
	\label{fig: block 1-toggle}
\end{figure}

The 1-toggle that lets the player carry a block through is more complicated, and is shown in Figure~\ref{fig: block 1-toggle carry}. If the player enters on the left with or without a block, they can get to the right as follows:
\begin{itemize}
	\item Move the top staircase to the right, so they can climb all the way down.
	\item Move the top staircase and then the bottom staircase to a single pile in the bottom left corner.
	\item Move the single pile to the bottom right corner.
	\item Use three blocks to build a staircase to the middle platform on the right, and move the rest of the blocks up to that platform.
	\item Use another three blocks to build a staircase to the right exit.
\end{itemize}

To reach either exit, there must be at least three blocks on the bottom level to form a staircase to the middle platform, and three blocks on the middle platform to form a staircase to the exit. In particular, six blocks must stay inside the gadget, so the player can leave with a block only if they brought one with them. If the player tries to enter the side opposite the one they most recently exited, they will be blocked by both staircases and unable to get across the gadget.

\begin{figure}
	\centering
	\includegraphics[width=\linewidth]{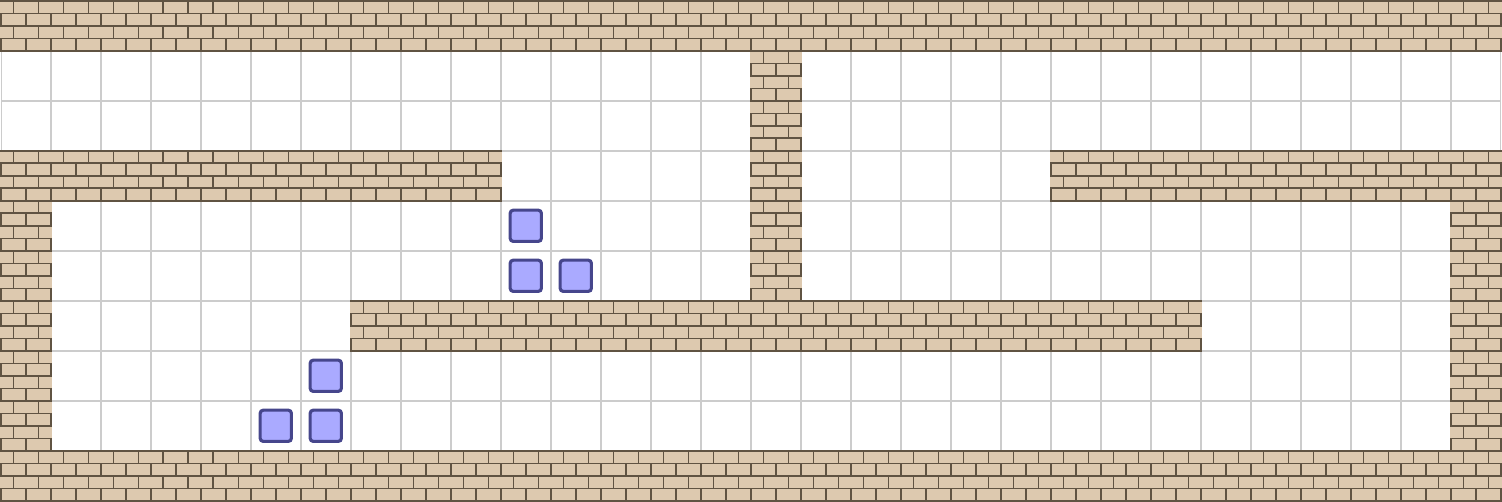}
	\caption{A 1-toggle for BlockDude that lets the player carry a block through it, currently traversable from left to right.}
	\label{fig: block 1-toggle carry}
\end{figure}

This 1-toggle might break if the player brings several additional blocks to it, but it will never be possible to bring more than one additional block because of the structure of our locking 2-toggles.

With these components, we can fill in our schematic for a locking 2-toggle (Figure~\ref{fig: block l2t compressed}), which we show in full in Figure~\ref{fig: block l2t full}. To summarize: the player can enter on either side, at the lower entrance. They can get to the block in the center, but must return to the side they came from. Then they can use this block to reach the top exit on the same side. This makes the center block inaccessible from the other side, so the other side cannot be traversed until the player comes back in the opposite direction and returns the center block.

\begin{figure}
	\centering
	\includegraphics[width=\linewidth]{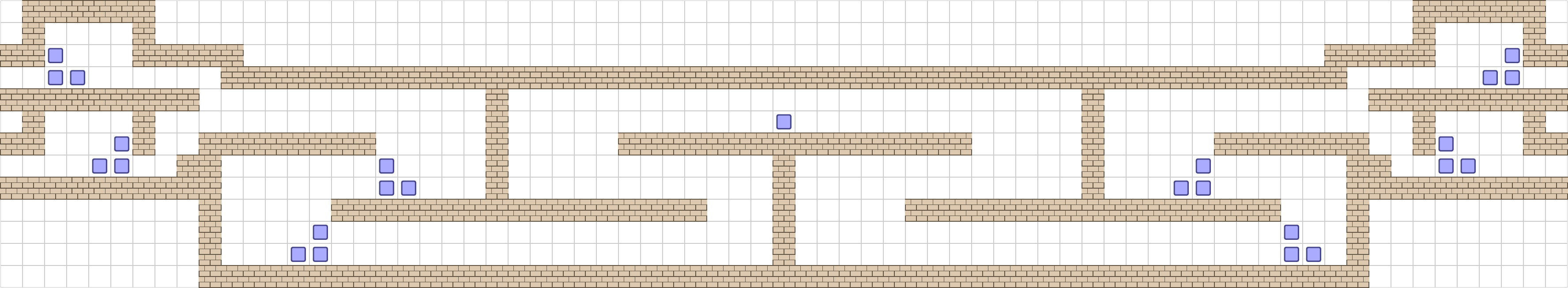}
	\caption{The full locking 2-toggle for BlockDude, combining Figures~\ref{fig: block l2t compressed}, \ref{fig: block 1-toggle}, and \ref{fig: block 1-toggle carry}.}
	\label{fig: block l2t full}
\end{figure}

\subsection{BloxDude is PSPACE-complete}
\label{sec:BloxDude}
In this section we discuss how to adapt the prior proof for BlockDude puzzles to work for blox which can both be picked up and pushed. All the valid traversals from our BlockDude constructions remain and we only need to prevent unwanted movement of the blox due to pushing. 

First, whenever there is a hallway in which a blox should not be able to be moved, such as all three hallways from the branching hallway, we add a step in the hallway, as shown in Figure~\ref{fig: no-blox-hallway}. Thus the blox cannot be carried and if it is pushed to the step it will become stuck.

\begin{figure}
	\centering
	\includegraphics[scale=1]{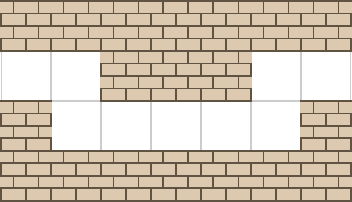}
	\caption{A blox cannot be moved through this hallway.}
	\label{fig: no-blox-hallway}
\end{figure}

Next we show how to adapt the 1-toggle with block traversal so it works in this setting. This is given in Figure~\ref{fig: blox 1-toggle carry}. The three-block-tall staircases ensure that bringing a single blox from the wrong direction does not allow deconstructing a staircase from behind. In particular, the middle layer has two blox in a row which cannot be pushed and thus one extra blox will not enable the Dude to deconstruct the staircase from that side.

\begin{figure}
	\centering
	\includegraphics[width=\linewidth]{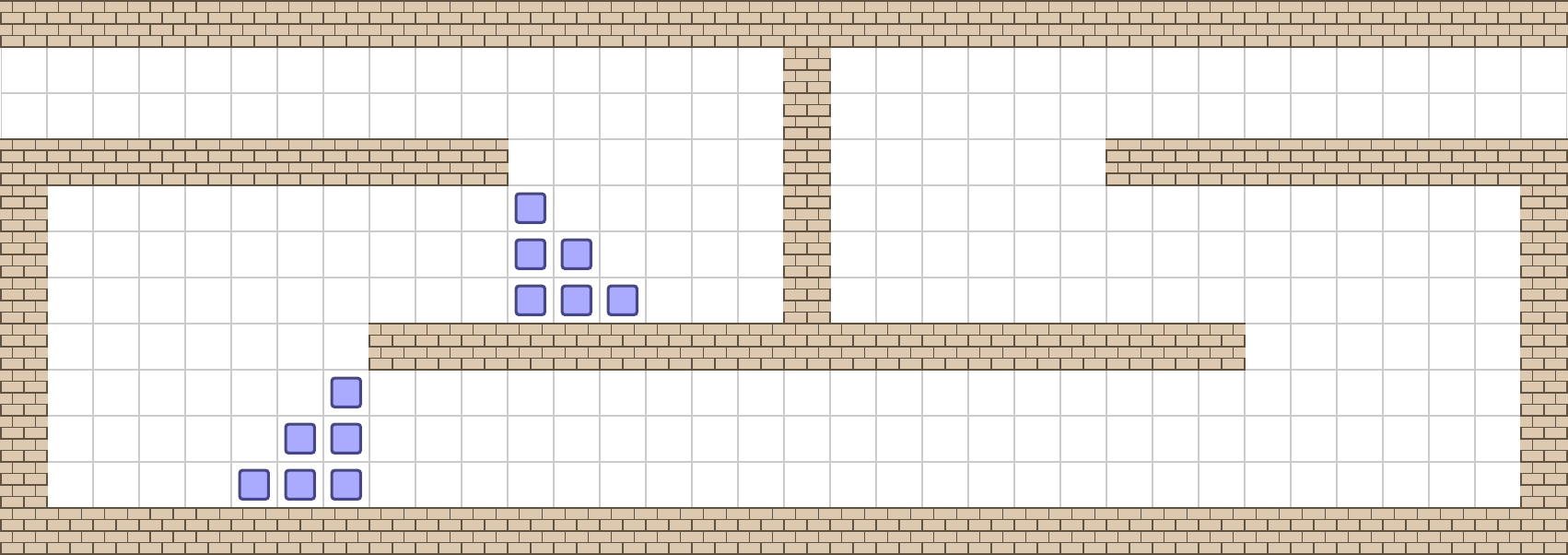}
	\caption{A 1-toggle for BloxDude that lets the player carry a block through it, currently traversable from left to right.}
	\label{fig: blox 1-toggle carry}
\end{figure}

We also need a regular 1-toggle, and the construction in Figure~\ref{fig: block 1-toggle} can be broken in the blox model. Luckily we have a hallway that prevents blox from being carried or pushed through it, so we can add such a hallway to each end of the gadget in Figure~\ref{fig: blox 1-toggle carry} preventing extra blox from entering or leaving.  This yields a regular 1-toggle which does not permit blox to pass through.

Once we have the prior two gadgets, it is clear the locking 2-toggle in Figure~\ref{fig: block l2t compressed} will still work in the blox model, giving the desired PSPACE-hardness result.

\section{Checkable Gadget Framework}
\label{sec:checkable gadget framework}

In this section, we introduce a new extension to the gadgets framework
which will be used in the rest of the paper.
This extension allows us to indirectly construct a gadget $G$
by first constructing a ``checkable'' version of $G$,
and then using ``postselection'' to obtain $G$.
The checkable $G$ behaves identically to $G$ except that
the agent can make undesired traversals into ``broken'' states
which prevent later ``checking'' traversals.
The postselection operation removes these possibilities by guaranteeing that the agent will perform the checking traversals at the end, so to solve reachability, the agent could never perform the undesired traversals.
The price we pay for this ability to constrain the behavior of gadgets is that the resulting simulations are no longer drop-in replacements as in the local simulations of Definition~\ref{def:local simulation}; instead we obtain ``nonlocal simulations'' which require altering the entire surrounding system of gadgets:

\begin{definition}
A finite set of gadgets \(\mathcal{G}\) \emph{[planarly] nonlocally simulates}
a gadget \(H\) if, for every finite set of gadgets \(\mathcal{G}'\),
there is a polynomial-time (many-one/Karp) reduction from [planar] reachability
with \(\{H\} \cup \mathcal{G}'\) to [planar] reachability with
\(\mathcal{G} \cup \mathcal{G}'\).
\end{definition}

Lemma~\ref{lem:simulation reduction} says that simulations are nonlocal simulations, so this notion is a generalization of Definition~\ref{def:local simulation}.

Next we define ``checkable'' gadgets via ``postselection'',
which transforms a gadget with broken states
(where a checking traversal sequence is impossible)
into an idealized gadget where those broken states are prevented.
At this stage, the prevention is by a magical force,
but we will later implement this force with a nonlocal simulation.

\begin{definition}\label{def:postselect}
Let $G$ be a gadget, $C$ be a traversal sequence on $L(G)$,
and $L'\subset L(G)$.
Call a state $q$ of $G$ \emph{broken} if $C$ is not legal from~$q$.
Assume that broken states are preserved by transitions on $L'$
in the sense that,
if $q$ is broken and there is a transition $(q,a)\to (q',b)$ where $a,b\in L'$,
then $q'$ is also broken.

Define \emph{$\PostSelect{G,C,L'}$} to be the gadget $G'$ where $L(G')=L'$, $Q(G')$ contains the nonbroken states of $G$, and $T(G')$ contains the transitions of $G$ restricted to $L'$ and $Q(G')$.\footnote{If every state of $G$ is broken,
then $\PostSelect{G,C,L'}$ has no states.
In this case, it is impossible to use $\PostSelect{G,C,L'}$ in a system of gadgets because that requires specifying an initial state,
so all of our theorems hold vacuously.
}
When there exist $C$ and $L'$ such that $\PostSelect{G,C,L'}$ is equivalent to $G'$, we say that $G$ is a \emph{checkable $G'$}, and we call $C$ the \emph{checking traversal sequence}.
\end{definition}

A traversal sequence $X$ is legal for $\PostSelect{G,C,L'}$ from state $q$ if and only if $XC$ is legal for $G$ from $q$, because both are equivalent to there being a nonbroken state reachable by traversing $X$.
Intuitively, $\PostSelect{G,C,L'}$ is the gadget that results from forcing the agent to traverse $C$ after solving reachability, to ensure that the gadget was left in a nonbroken state, and hiding locations in $L \setminus L'$.
$\PostSelect{G,C,L'}$ behaves like $G$ on the locations $L'$ except that transitions into broken states are prohibited.

We now state the main result of the checkable gadget framework, which is in
terms of two simple (and often easy-to-implement) gadgets
SO (single-use opening) and MSC (merged single-use closing gadgets)
defined in Section~\ref{sec:Base Gadgets}.

\begin{theorem}
  \label{thm:postselect}
  For any \(G\), \(C\), and \(L'\) satisfying the assumptions of Definition~\ref{def:postselect}, \(\{G, \SO, \MSC\}\) planarly nonlocally simulates \(\PostSelect{G, C, L'}\).
\end{theorem}

The goal of this section is to prove Theorem~\ref{thm:postselect}.
Figure~\ref{fig: postselection overview} provides a schematic overview of the gadget simulations throughout this section that culminate in this result.
In Section~\ref{sec:Base Gadgets}, we describe the base gadgets needed for our
construction.
In Section~\ref{sec:Nonlocal}, we prove that nonlocal simulations compose in the natural way.
In Section~\ref{sec:simply checkable}, we introduce a particularly simple kind of checkable gadget, and show that they nonlocally simulate the gadget they are based on.
Finally, in Section~\ref{sec:Postselected Gadgets} we use all of these tools to prove Theorem~\ref{thm:postselect}.

\begin{figure}
  \centering
  \includegraphics{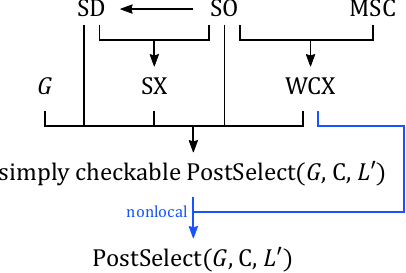}
  \caption{Overview of gadget simulations used for postselection.
    Black arrows show local simulations and blue arrows show nonlocal simulations.}
  \label{fig: postselection overview}
\end{figure}

\subsection{Base Gadgets}
\label{sec:Base Gadgets}

We now define two base gadgets and three additional derived gadgets,
shown in Figure~\ref{fig:v2 gadget states},
that we use to implement the machinery of checkable gadgets.
All five of these gadgets can change state only a bounded number of times and thus reachability with only these gadgets is in \NP;
they are ``LDAG'' in the language of \cite{lynch2020framework}.

\begin{figure}
	\centering
  \def\scale{0.8}
  \def\width{2.9cm}
	\subcaptionbox{\label{fig:SO_state}\centering Single-use opening gadget (SO)}[\width]{\includegraphics[scale=\scale]{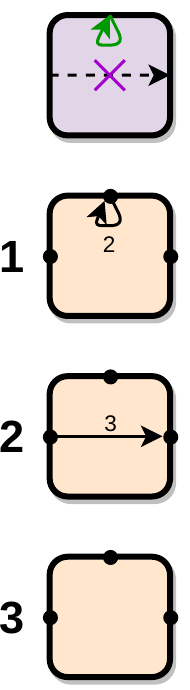}}
	\hfill
	\subcaptionbox{\label{fig:MSC_state}\centering Merged single-use closing gadget (MSC)}[\width]{\includegraphics[scale=\scale]{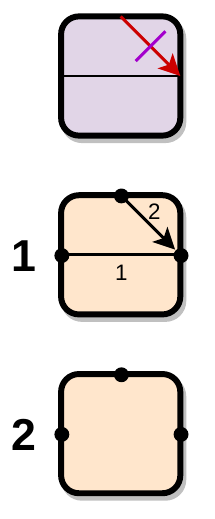}}
	\hfill
	\subcaptionbox{\label{fig:SD_state}\centering Dicrumbler/\allowbreak single-use diode (SD)}[\width]{\includegraphics[scale=\scale]{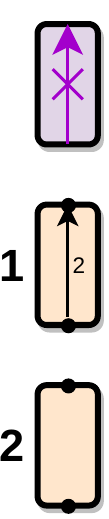}}
  \hfill
	\subcaptionbox{\label{fig: single-use-crossover_state}\centering Single-use crossover (SX)}[\width]{\includegraphics[scale=\scale]{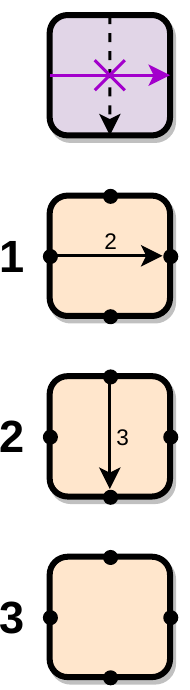}}
  \hfill
	\subcaptionbox{\label{fig: weak-closing-crossover_state}\centering Weak closing crossover (WCX)}[\width]{\includegraphics[scale=\scale]{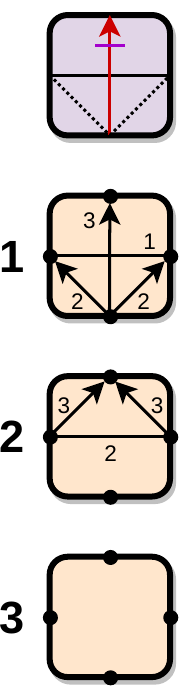}}
	\caption{Icons (top) and state diagrams (bottom)
    for two base gadgets (a--b) and three derived gadgets (c--e).
    Green arrows show opening traversals,
    red arrows show closing traversals,
    and purple crosses indicate traversals that close themselves.}
	\label{fig:v2 gadget states}
\end{figure}

The two base gadgets required for our construction are shown in
Figure~\ref{fig:SO_state}--\ref{fig:MSC_state}:

\begin{enumerate}[label=(\alph*)]
\item The \emph{single-use opening (SO)} gadget, shown in Figure~\ref{fig:SO_state}, is a three-state three-location gadget. In state 1, the ``opening'' location has a self-loop traversal (also called a button, or a port in \cite{doors}), which transitions to state 2. State 2 allows a single traversal between the other two locations, after which (in state 3) no traversals are possible.
\item The \emph{merged single-use closing (MSC)} gadget, shown in Figure~\ref{fig:MSC_state}, is a two-state three-location gadget. In the ``open'' state 1, horizontal traversals in both directions are freely available.  After a traversal from top to right, the gadget transitions to the ``closed'' state 2, where no traversals are possible.
\end{enumerate}

Next we describe three useful gadgets for our construction
which can be built from these base gadgets.

The \emph{dicrumbler/single-use diode (SD)} gadget,
shown in Figure~\ref{fig:SD_state}, is a two-state two-location gadget.
In state 1, there is a single directed traversal between the two locations,
which permanently closes the gadget in state 2 where no traversals
are possible.
The SD gadget can be simulated by either of the two base gadgets:
it is equivalent to state 2 of SO,
and to MSC restricted to the two locations incident to the closing traversal.

The \emph{single-use crossover (SX)} gadget, shown in Figure~\ref{fig: single-use-crossover_state}, allows one traversal from left to right and then one from top to bottom. It can be simulated using SO and SD gadgets as shown in Figure~\ref{fig: single-use-crossover}. The top location in the simulation cannot be entered until the top SO is opened. This opening is possible only after traversing the first two SDs, which prevents any further traversals coming from the left or going to the right. The bottom SO prevents premature traversals going to the bottom.

\begin{figure}
  \centering
	\includegraphics{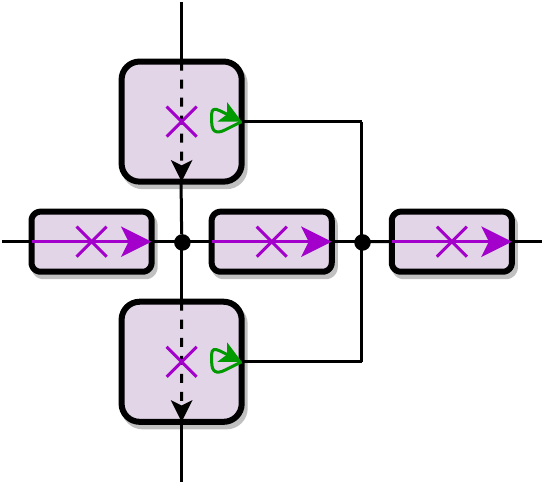}
	\caption{Construction of the single-use crossover from SO and SD gadgets.}
	\label{fig: single-use-crossover}
\end{figure}

The \emph{weak closing crossover (WCX)},
shown in Figure~\ref{fig: weak-closing-crossover_state},
initially allows traversals freely between the left and right.
If a bottom-to-top traversal is performed, no more traversals are possible.
However, a bottom-to-left or bottom-to-right traversal is also possible
(which also opens up left-to-top or right-to-top traversals),
making the crossover ``leaky''.
The weak closing crossover can be simulated using SO, MSC, and SD gadgets,
as shown in Figure~\ref{fig: weak-closing-crossover}.
To open the upper-right SO,
the agent needs to traverse the upper-left SO and then close the middle MSC.
To open the upper-left SO, the agent will need to close the leftmost MSC.
Having closed both the left and the middle MSCs, the agent is forced to traverse the bottom SO and close the rightmost MSC. 
The bottom SO can only be opened by
the agent traversing entering the bottom and traversing bottom two SDs, preventing any future traversals from the bottom.
In summary, in order to exit the top, the agent must have entered the bottom in the past, and have closed all three MSCs.
Entering the bottom changes to state 2, and exiting the top changes to state 3.

\begin{figure}
  \centering
	\includegraphics{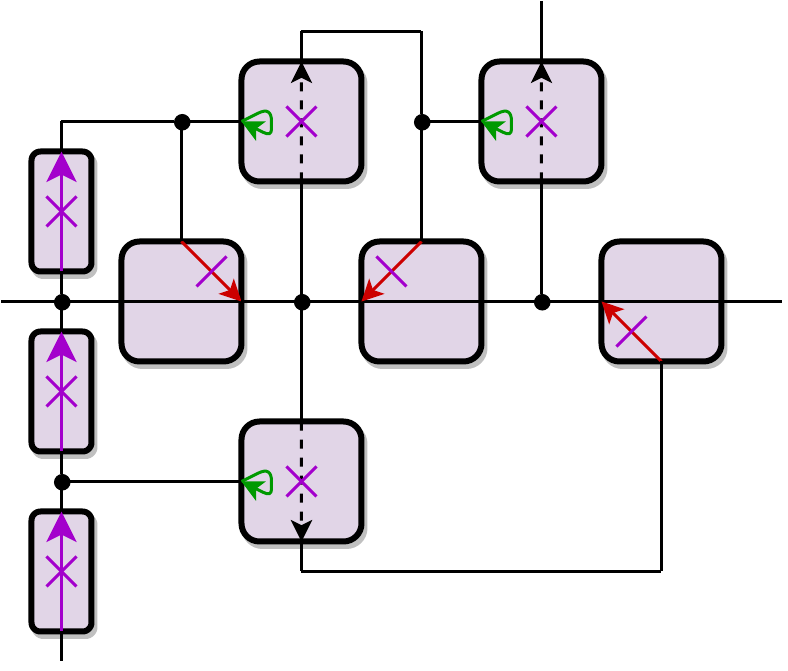}
	\caption{Construction of the weak closing crossover from SD, SO, and MSC gadgets.}
	\label{fig: weak-closing-crossover}
\end{figure}

\subsection{Nonlocal Simulation Composition}
\label{sec:Nonlocal}

A crucial fact about nonlocal simulation is that nonlocal simulations can be composed:

\begin{lemma}\label{lem:compose nonlocal}
Let $\mathcal G$ and $\mathcal H$ be finite sets of gadgets. Suppose $\mathcal G$ [planarly] nonlocally simulates every gadget in $\mathcal H$, and $\mathcal H$ [planarly] nonlocally simulates another gadget $H$. Then $\mathcal G$ [planarly] nonlocally simulates $H$.
\end{lemma}
\begin{proof}
  For a finite set of gadgets $\mathcal G'$,
  we must find a polynomial-time reduction from reachability with
  $\{H\}\cup\mathcal G'$ to reachability with $\mathcal G\cup\mathcal G'$.
  Let $\mathcal H=\{H_1,\dots,H_n\}$, where $n=|\mathcal H|$,
  and let $\mathcal H_i$ be the prefix $\{H_1,\dots, H_i\}$,
  so $\mathcal H_n=\mathcal H$.
  Then we construct a chain of reductions
  between reachability with different sets of gadgets:
  $$\{H\}\cup\mathcal G'
    ~\to~\mathcal G\cup\mathcal H_n\cup \mathcal G'
  ~\to~\mathcal G\cup\mathcal H_{n-1}\cup\mathcal G'
  ~\to~\cdots\to\mathcal G\cup\mathcal H_1\cup\mathcal G'
  ~\to~\mathcal G\cup\mathcal G'.$$
    The first reduction is because $\mathcal H = \mathcal H_n$
  nonlocally simulates~$H$.
    The remaining reductions come from the assumption that $\mathcal G$
  nonlocally simulates each $H_i \in \mathcal H$,
  which implies that there is a polynomial-time reduction from reachability
  with $\{H_i\}\cup\mathcal G\cup\mathcal H_{i-1}\cup\mathcal G'
  =\mathcal G\cup\mathcal H_i\cup\mathcal G'$ to reachability with
  $\mathcal G\cup\mathcal G\cup\mathcal H_{i-1}\cup\mathcal G'
  =\mathcal G\cup\mathcal H_{i-1}\cup\mathcal G'$.
\end{proof}

\subsection{Simply Checkable Gadgets}
\label{sec:simply checkable}

Next, we define a special kind of checkable gadgets, called ``simply checkable'' gadgets. A simply checkable $G$ is essentially a checkable $G$ where the checking sequence consists of a single traversal between two locations not in $L(G)$, called $\cIn$ and $\cOut$. Simply checkable gadgets will be useful for an intermediate step in our proof of Theorem~\ref{thm:postselect}.

\begin{definition} \label{def:simplycheckable}
For a gadget \(G\), a \emph{simply checkable} $G$ is a gadget $G'$ satisfying the following properties:
\begin{enumerate}
\item \(L(G') = L(G) \sqcup \{\cIn, \cOut\}\) has two new locations $\cIn, \cOut$. For planar gadgets, the cyclic orderings of the shared locations $L(G)$ are the same. (Locations $\cIn$ and $\cOut$ can be added to the cyclic order anywhere.)
\item There is a function $f:Q(G)\to Q(G')$ assigning a state of $G'$ to each state of $G$.
\item \label{def:simplycheckable:can-extend}
For any traversal sequence \(X\) that is legal for \(G\) from state $q$, the concatenated traversal sequence \(X \cdot [\cIn \to \cOut]\) is legal for \(G'\) from $f(q)$. 
\item \label{def:simplycheckable:end}
Every traversal sequence that ends at \(\cOut\) and is legal for $G'$ from state $f(q)$ has the form
\[
  X \cdot [\cIn \to \bullet, \bullet \to \bullet, \dots, \bullet \to \cOut]
\]
where \(X\) is legal for $G$ from state $q$
and the omitted $\bullet$ locations (if any) belong to $L(G)$.
\end{enumerate}
\end{definition}

Intuitively, a simply checkable $G$ in state $f(q)$ behaves the same as \(G\) does in state $q$, provided that afterward the agent performs a traversal sequence from \(\cIn\) to \(\cOut\) (which may involve the agent exiting and re-entering the gadget, but only via nonchecking locations).
The gadget can do essentially anything in a traversal sequence not ending in $\cOut$.

Any simply checkable $G$ is also a checkable $G$: if $G'$ is a simply checkable $G$, then $\PostSelect{G',[\cIn\to\cOut],L(G)}$ is equivalent to $G$.

We show that a simply checkable $G$ can nonlocally simulate $G$
while preserving planarity, using an auxiliary gadget.
First, define the $\hallway$ gadget to be the one-state two-location gadget
with transitions in both directions between the locations
(i.e., a ``branching hallway'' with only two locations).
A \emph{checkable hallway crossover} is a simply checkable hallway
where the added locations \(\cIn\) and \(\cOut\)
are not adjacent in the cyclic order,
i.e., they interleave with the two hallway locations.
For example, the weak closing crossover
from Figure~\ref{fig: weak-closing-crossover_state}
is a checkable hallway crossover,
where the horizontal traversal corresponds to the hallway,
the bottom location is $\cIn$, and the top location is $\cOut$.

\begin{lemma}
  \label{lem:checkable nonlocal}
  Let \(G'\) be a simply checkable $G$ and let \(\CHX\) be a checkable hallway crossover.
  Then
  \begin{enumerate}
  \item \(\{G'\}\) nonlocally simulates \(G\); and
  \item \(\{G', \CHX\}\) planarly nonlocally simulates \(G\).
  \end{enumerate}
\end{lemma}
\begin{proof}
  For any gadget set $\mathcal G'$,
  we construct a polynomial-time reduction
  from reachability with $\{G\}\cup\mathcal G'$
  to reachability with $\{G'\}\cup\mathcal G'$,
  or from planar reachability with $\{G\}\cup\mathcal G'$
  to planar reachability with $\{G', \CHX\}\cup\mathcal G'$.
  Suppose we have a [planar] system $S$ of gadgets from $\{G\}\cup\mathcal G'$,
  along with a designated starting location \(s\) and target location \(t\).
  Let $G_1, \dots, G_n$ denote the copies of $G$ in~$S$,
  and let $q_1, \dots, q_n$ be their respective initial states in~$S$.
  We build a new system \(S'\) of gadgets from $\{G'\}\cup\mathcal G'$ as follows; refer to Figure~\ref{fig:checkable nonlocal}.
  \begin{enumerate}
  \item Replace each copy $G_i$ of gadget $G$ with initial state $q_i$ in $S$
    by a corresponding copy $G'_i$ of $G'$ with initial state $f(q_i)$,
    whose copies of $\cIn$ and $\cOut$ are named $\cIni i$ and $\cOuti i$.
  \item Connect \(t\) to \(\cIni1\).
    In the planar case, we place a copy of $\CHX$ on each crossing this creates, with the check line on the way from $t$ to $\cIni1$.
  \item Connect \(\cOuti i\) to \(\cIni{i+1}\) for each $i$.
    In the planar case, we place a copy of $\CHX$ on each crossing this creates, with the check line on the way from $\cOuti i$ to $\cIni{i+1}$.
  \end{enumerate}
  Our reduction outputs this new system $S'$ along with
  the same start location $s$ and
  the new target location $t' = \cOuti n$.

  \begin{figure}
    \centering
    \includegraphics{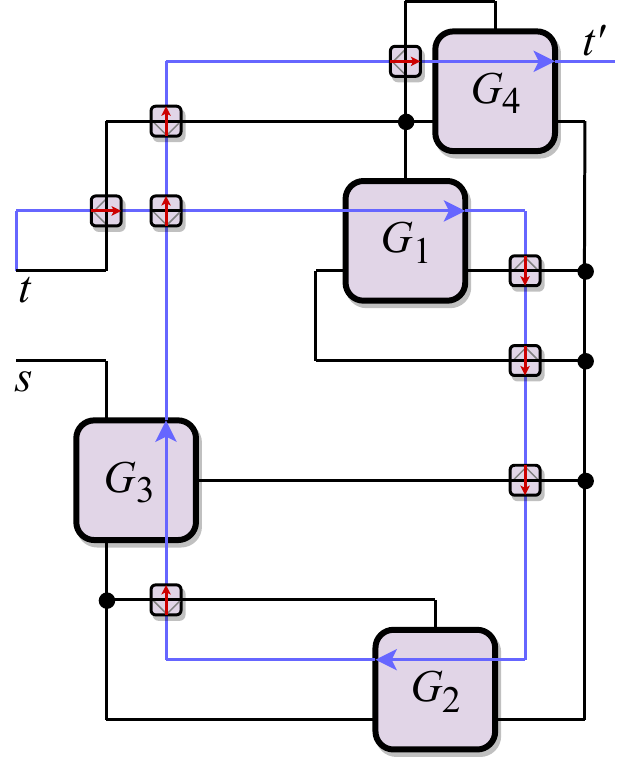}
    \caption{Our nonlocal simulation for the proof of Lemma~\ref{lem:checkable nonlocal}. The system is modified by replacing each copy of $G$ with a copy of $G'$ and adding the blue path from $t$ through $\cIn\to\cOut$ on each one.}
    \label{fig:checkable nonlocal}
  \end{figure}

  This construction clearly takes polynomial time.
  To prove that the reduction is valid, we must show that
  there is a legal system traversal $s \to^* \cOuti n$ in \(S'\)
  if and only if
  there is a legal system traversal $s \to^* t$ in \(S\).

  First suppose there is a legal system traversal $s \to^* t$ in $S$.
  Then this solution can be extended to a legal system traversal
  $s \to^* \cOuti n$ in $S'$ by appending the traversal
  \(\cIni i \to \cOuti i\) on $G'_i$
  for each $i$ in increasing order,
  and in the planar case,
  adding the needed traversals of the inserted copies of $\CHX$
  (including the check traversals needed to get from $t$ to $\cIni 1$ and
  from each $\cOuti i$ to $\cIni{i+1}$).
  The appended $\cIni i \to \cOuti i$ traversals
  are all valid because Property~\ref{def:simplycheckable:can-extend} of
  Definition~\ref{def:simplycheckable} requires that
  any legal traversal sequence for \(G\) can be extended by
  \(\cIn \to \cOut\) to yield a legal traversal sequence for \(G'\).
  For the same reason, the appended $\cIn \to \cOut$ traversals
  in copies of $\CHX$ are valid.
  Also, the inserted hallway traversals of the copies of $\CHX$
  are all valid from the definition of checkable hallway crossover,
  because they occur before all appended $\cIn \to \cOut$ traversals.

  Now suppose that there is a legal system traversal $s \to^* \cOuti n$ in $S'$.
  Define $\pIni i, \pOuti i$ to be the check in and out locations
  for all checkable gadgets (copies of both $G'$ and $\CHX$), in the order that
  these check traversals occur in the intended solution described above.
  By Property~\ref{def:simplycheckable:end} of Definition~\ref{def:simplycheckable},
  the agent can only exit the $i$th checkable gadget ($G'$ or $\CHX$)
  at \(\pOuti i\)
  if it previously entered at the corresponding \(\pIni i\).
  In $S'$, the only location connected to $\pIni{i+1}$ is $\pOuti i$
  (ignoring hallway traversals of $\CHX$ gadgets),
  so this property implies that $\cOuti i$ was previously visited as well.
  By induction, the solution must have reached $\pIni 1$ via \(t\),
  and then traversed all of the \(\pIni i\) and \(\pOuti i\) locations
  (possibly with some detours).
    Consider the prefix $X'$ of the solution up to the first time $t$ is visited,
  and let $X$ be the modification to remove any hallway traversals
  of the copies of $\CHX$.
  We claim $X$ is a solution for $S$.
  Clearly $X$ is a system traversal $s \to^* t$
  and satisfies all unmodified gadgets (from $\mathcal G'$).
  By Property~\ref{def:simplycheckable:end} of Definition~\ref{def:simplycheckable},
  $\pIni i$ and $\pOuti i$ are visited at most once in the full solution,
  and the prefix of the solution prior to visiting $\pIni i$ is legal for
  the $i$th checked gadget.
    Because each $\pIni i$ is visited after $t$, it is not visited in $X$,
  and thus $X$ is legal for $G_i$.
  Similarly, $X$ makes only hallway traversals of $\CHX$, so removing those
  traversals is valid in $S$ where there were direct connections
  before the crossings were introduced.
  Therefore $X$ is a valid system traversal $s \to^* t$ in~$S$.
\end{proof}

\subsection{Postselected Gadgets}
\label{sec:Postselected Gadgets}

We now finally prove our main result, Theorem~\ref{thm:postselect}:
postselection can be achieved using only the two base gadgets
from Section~\ref{sec:Base Gadgets}, while preserving planarity.

It will be convenient to assume all of our gadgets are \emph{transitive}:
if there are two transitions $(q_1,\ell_1)\to(q_2,\ell_2)\to(q_3,\ell_3)$,
then there is also a transition $(q_1,\ell_1)\to(q_3,\ell_3)$.
For reachability, this makes no difference:
we can replace any gadget with its transitive closure without affecting the answers to any reachability problems,
since we can always think of the transition $(q_1,\ell_1)\to(q_3,\ell_3)$ as a sequence of two transitions.
That is, every gadget is equivalent for reachability to some transitive gadget,
and in particular there are nonlocal simulations in both directions.

\begin{proof}[Proof of Theorem~\ref{thm:postselect}]
  Assume without loss of generality that $G$ is transitive,
  by replacing $G$ with its transitive closure.

  We will show that $\{G,\SO,\MSC,\SD,\SX,\WCX\}$ planarly \textit{locally}
  simulates some gadget $G'$ which is a simply checkable $\PostSelect{G,C,L'}$.
  As shown in Section~\ref{sec:Base Gadgets}
  (Figures~\ref{fig: single-use-crossover}
  and~\ref{fig: weak-closing-crossover} in particular),
  $\{\SO,\MSC\}$ planarly locally simulates $\WCX$, $\SX$, and $\SD$.
  By combining these local simulations, we obtain that
  $\{G,\SO,\MSC\}$ planarly locally simulates the same $G'$.
  By Lemma~\ref{lem:simulation reduction}, this is also a nonlocal simulation.
  By Lemma~\ref{lem:checkable nonlocal},
      for any checkable hallway crossover gadget $\CHX$,
  $\{G',\CHX\}$ planarly nonlocally simulates~$G'$.
  Because $\{\SO,\MSC\}$ planarly simulates the weak closing crossover
  (Figure~\ref{fig: weak-closing-crossover}),
  which is a checkable hallway crossover,
  it follows from Lemma~\ref{lem:compose nonlocal} that
  $\{G,\SO,\MSC\}$ planarly nonlocally simulates $\PostSelect{G,C,L'}$,
  proving the theorem.

  Now we show that $\{G,\SO,\MSC,\SD,\SX,\WCX\}$ planarly locally simulates
  some gadget $G'$ which is a simply checkable $\PostSelect{G,C,L'}$.
  Unpacking the definitions of ``simply checkable'' and \textsf{Postselect},
  we must simulate a gadget \(G'\) that satisfies the following properties:
  \begin{enumerate}
  \item $L(G')=L'\sqcup\{\cIn,\cOut\}$.
  \item There is a function $f$ from unbroken states of $G$ to states of $G'$.
  \item For any traversal sequence $X$ on $L'$, if \(XC\) is legal for $G$ from state $q$, then \(X \cdot [\cIn \to \cOut]\) is legal for $G'$ from state $f(q)$.
  \item Any traversal sequence that ends with \(\cOut\) and is legal for \(G'\) from state $f(q)$ has the form \(X \cdot [\cIn \to \bullet, \bullet \to \bullet, \dots, \bullet \to \cOut]\), where \(X\) is a traversal sequence on $L'$, \(XC\) is legal for $G$ from state $q$, and all the omitted $\bullet$ locations are in \(L'\).
  \end{enumerate}
    We construct our simulation of the gadget \(G'\) starting from \(G\)
  as follows; refer to Figure~\ref{fig: postselection}.
  \begin{enumerate}
  \item For purposes of description,
    orient so that $G$ has all of its locations on the top of its bounding box.
    We will place the locations for the simulated gadget
    on a horizontal line $L$ above $G$
    (so they will lie on the outside face).
  \item For each location \(l \in L'\),
    add a long upward edge \(e_l\) connecting \(l\) in \(G\)
    to a new location $l'$ on $L$.
        Because the edges are all vertical, they do not cross each other, and
    the $l'$ locations appear in the same cyclic (left-to-right) order
    as $l \in L'$.
  \item Place \(\cIn\) on $L$ left of all \(e_l\) edges.
    Starting from \(\cIn\), draw a non-self-crossing path
    that crosses each of the \(e_l\) in one rightward pass,
    then turn down,
    then cross each \(e_l\) a second time in one leftward pass
    in between the first pass and~$G$.
    We ensure any further crossings with the edges \(e_l\)
    take place between these two delimiter passes,
    which we call the top and bottom delimiters,
    by routing paths across the bottom delimiter before crossing any \(e_l\).
    These delimiters serve to ``cut off'' the rest of the construction, preventing leakage.
  \item For each traversal \(a_i \to b_i\) in the sequence \(C = [a_1 \to b_1,\dots,a_k \to b_k]\),
    add a single-use opening gadget $O_i$ and a dicrumbler $D_i$,
    near locations $b_i$ and $a_i$ respectively.
    Connect the opening location of $O_i$ to the entrance of $D_i$
    (routing up across the bottom delimiter, then horizontally, then down).
    Connect the exit of $D_i$ to $a_i$,
    and connect $b_i$ to the entrance of $O_i$.
  \item Connect the exit of each \(O_i\) to the opening location of \(O_{i+1}\),
    routing up across the bottom delimiter, then all the way left,
    then up, then right, then down.
  \item Finally, connect \(\cIn\) to the opening location of \(O_1\)
    after the two delimiter passes; and
    connect the exit of \(O_k\) to \(\cOut\),
    routing up across the bottom delimiter, then all the way left, then up.
  \end{enumerate}

  We call the path we have constructed from $\cIn$ to $\cOut$ the \emph{checking path}. For an unbroken state $q$ of $G$, the corresponding state $f(q)$ of $G'$ is simulated by placing $G$ in state $q$ and all other gadgets in their usual initial states.

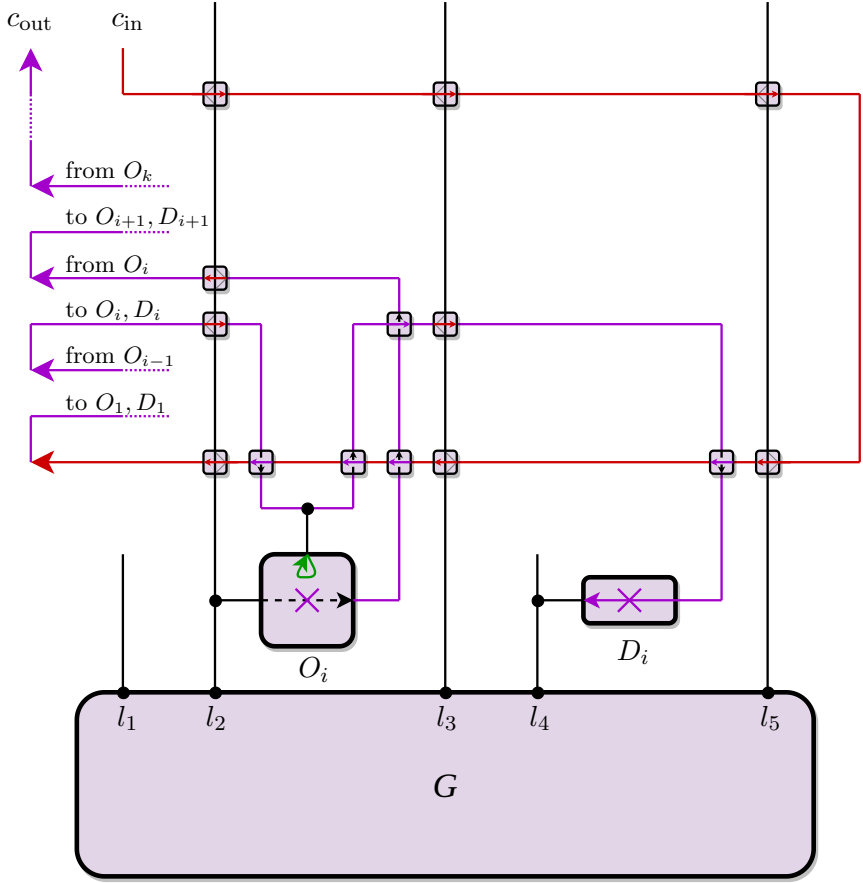
\begin{figure}[t]
  \centering
  \begin{overpic}[width=0.7\linewidth]{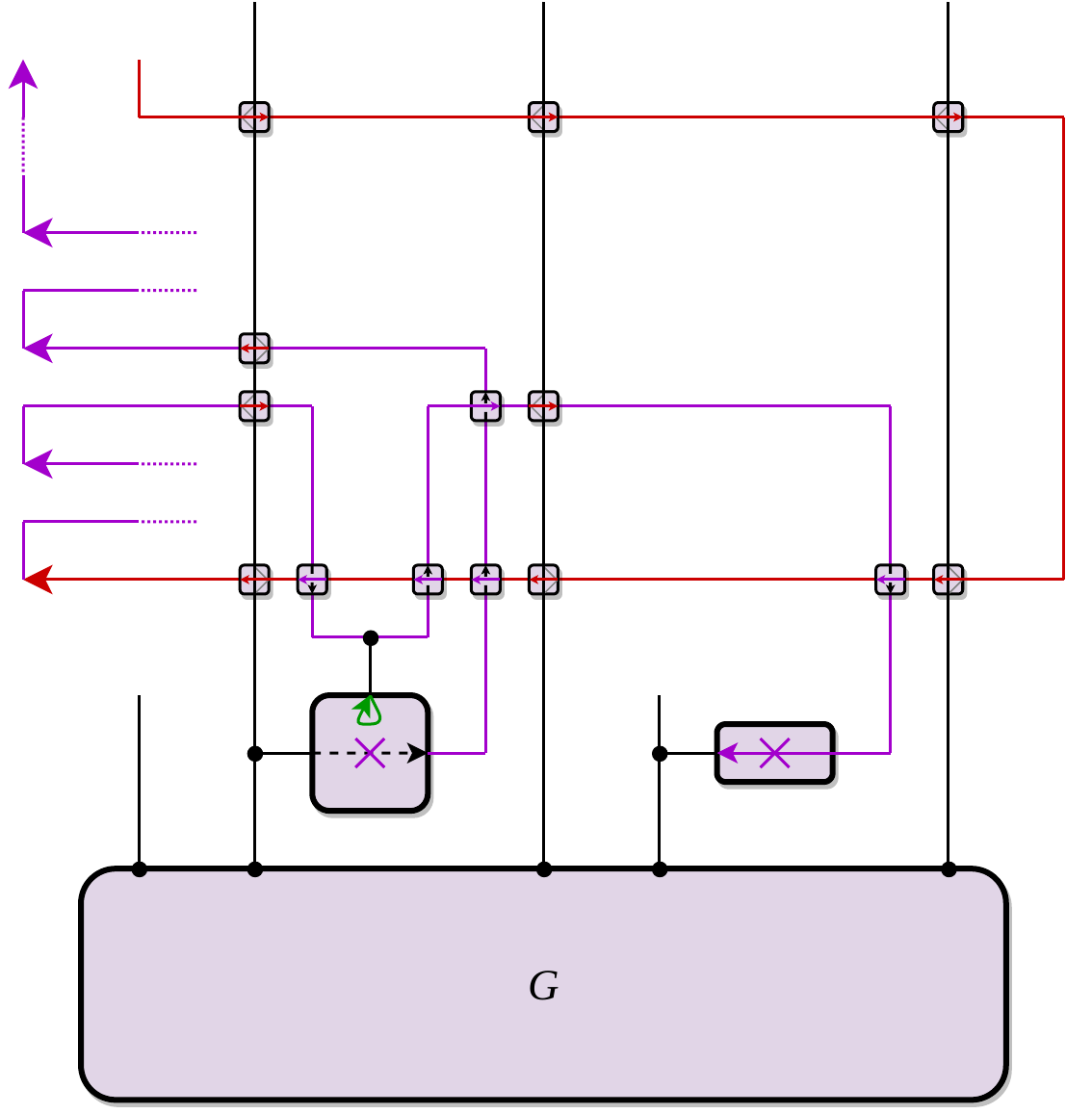}
    \put(13,98){\makebox(0,0)[c]{\strut \(\cIn\)}}
    \put(2,98){\makebox(0,0)[c]{\strut \(\cOut\)}}
    \put(13,19){\makebox(0,0)[c]{\strut \(l_1\)}}
    \put(23,19){\makebox(0,0)[c]{\strut \(l_2\)}}
    \put(49,19){\makebox(0,0)[c]{\strut \(l_3\)}}
    \put(59.5,19){\makebox(0,0)[c]{\strut \(l_4\)}}
    \put(85.5,19){\makebox(0,0)[c]{\strut \(l_5\)}}
    \put(34,24.5){\makebox(0,0)[c]{\strut \(O_i\)}}
    \put(70,26.5){\makebox(0,0)[c]{\strut \(D_i\)}}
    \put(6,54.5){\makebox(0,0)[l]{\footnotesize to $O_1, D_1$}}
    \put(6,59.75){\makebox(0,0)[l]{\footnotesize from $O_{i-1}$}}
    \put(6,65){\makebox(0,0)[l]{\footnotesize to $O_i, D_i$}}
    \put(6,70.25){\makebox(0,0)[l]{\footnotesize from $O_i$}}
    \put(6,75.5){\makebox(0,0)[l]{\footnotesize to $O_{i+1}, D_{i+1}$}}
    \put(6,80.75){\makebox(0,0)[l]{\footnotesize from $O_k$}}
  \end{overpic}
  \caption{The simulation of a simply checkable, postselected version of the gadget \(G\).  The two initial crossings of the edges \(e_l\) connecting locations in \(L'\) to the outside are shown in red.  The rest of the checking path is shown in purple.  All further crossings of the checking path with edges \(e_l\) occur between the two initial crossings.  In this example, \(L = \{l_1, l_2, l_3, l_4, l_5\}\) and \(L' = \{l_2, l_3, l_5\}\).  The \(i\)th checking traversal \([l_4 \to l_2]\) is enforced by \(O_i\) and \(D_i\).}
  \label{fig: postselection}
\end{figure}

  This construction is nonplanar in two ways: our new checking path crosses the edges \(e_l\) and also crosses itself.  In the former case we replace the crossing with a weak closing crossover, oriented so that the checking path closes \(e_l\).  In the latter case we replace the crossing with a single-use crossover, oriented correctly so that the agent can traverse the two directions in the expected order detailed below.  We must prove this construction has the properties stated above. By construction, its locations are $L'\sqcup\{\cIn,\cOut\}$.

  Suppose \(XC\) is legal for $G$ from state $q$.  We can perform \(X \cdot [\cIn \to \cOut]\) in the simulation where $G$ starts in $q$ by first performing \(X\) in the natural way (using the edges \(e_l\)) and then following the checking path: starting at \(\cIn\), for each \(i\) we visit the opening location of \(O_i\), then go through \(D_i\), then traverse \(a_i \to b_i\) via \(G\), then traverse \(O_i\).  This path brings us to \(\cOut\) at the end, and its restriction to \(G\) is exactly \(XC\).

  Now suppose that there is a legal traversal sequence for \(G'\) from state $f(q)$ ending in \(\cOut\).  Putting ourselves in the position of a forgetful agent, we find ourselves at \(\cOut\) and must determine how we got there.
  We can induct backwards along the checking path (as in the proof of Lemma~\ref{lem:checkable nonlocal}) to show that we must have visited \(\cIn\), using the facts that in order to exit the closing side of a weak closing crossover we must have entered it on the opposite side, and that in order to exit from \(O_i\) we must have visited its opening location.

  Thus at some point in the path we entered \(G'\) through \(\cIn\), crossed all the \(e_l\) twice, and then for every \(a_i \to b_i\) of \(C\) in order we opened \(O_i\), traversed \(D_i\), and later traversed \(O_i\). Crossing each $e_l$ twice closes the weak closing crossovers, making $e_l$ no longer traversable.  Between traversing \(D_i\) and \(O_i\), we somehow must have gotten from \(a_i\) to \(b_i\).  We cannot have used the edges \(e_l\) because they were already closed during the initial crossings.  
  So we must have made transitions only in $G$, of the form $(q_1,\ell_1=a_i)\to(q_2,\ell_2)\to\dots\to(q_k,\ell_m=b_i)$. 
  Since $G$ is transitive, we could equivalently have made the single transition $(q_1,a_i)\to(q_k,b_i)$,
  and in particular have traversed $a_i\to b_i$.

Similarly, after the initial two crossings of the \(e_l\), we can't have left this simulated gadget or entered \(G\) except for the traversals of \(C\).  Finally, we take advantage of the fact that before entering \(\cIn\), the simulation behaves exactly like \(G\) except that only locations in $L'$ are accessible.  So the full path through the simulation $G'$ ending at $\cOut$ must have the following form:
\begin{enumerate}
\item We use \(G'\) as if it were \(G\) (restricted to the locations of \(L'\)) with initial state $q$, performing some traversal sequence $X$.
\item We enter \(G'\) through \(\cIn\).
\item We possibly leak out of \(G'\) or into \(G\) via locations in \(L'\), through the weak closing crossovers at the initial two crossings with each \(e_l\). Call the sequence of traversals made during this phase $Y$.
\item Eventually, we finish all of initial crossings with $e_l$, and moved to the $O_i$s and $D_i$s.
\item We perform the traversal sequence \(C\) in \(G\) without any additional traversals in $G$ in between and without leaving \(G'\).
\item Finally, we leave \(G'\) through \(\cOut\).
\end{enumerate}
Therefore the sequence of traversals on $G'$ has the form \(X \cdot [\cIn \to \bullet, \bullet \to \bullet, \dots, \bullet \to \cOut]\) and the sequence of traversals just on $G$ is $XYC$, where $X$ and $Y$ are traversal sequences on $L'$ and the omitted $\bullet$ locations are in $L'$. In particular, $XYC$ is legal for $G$ from state $q$, so by the assumption that broken states are preserved by transitions on $L'$, $XC$ is legal for $G$ from $q$. This is the final condition we needed, so $G'$ is a simply checkable $\PostSelect{G,C,L'}$.
\end{proof}

\section{BoxDude is PSPACE-complete}
\label{sec:BoxDude}

We now show that BoxDude is PSPACE-complete via a reduction from reachability with nondeterministic locking 2-toggles.  In this model, boxes can be pushed horizontally by the Dude but cannot be picked up.  We will make use of the postselection construction from Section~\ref{sec:checkable gadget framework} in order to nonlocally simulate nondeterministic locking 2-toggles.

Similarly to BlockDude we must build a branching hallway in order to connect the locations of our gadgets.  This time, we also build a directed crossover gadget.
These gadgets are shown in Figure~\ref{fig: box hallways}. Directed crossovers can be used to construct undirected crossovers as in Figure~\ref{fig: directed crossover}.
This allows us to connect locations in nonplanar ways, and reduce from reachability instead of planar reachability.  We note a diode gadget is easy to build by simply having a height 2 drop.

\begin{figure}
  \centering
  \subcaptionbox{Branching Hallway}{\includegraphics[scale=0.7]{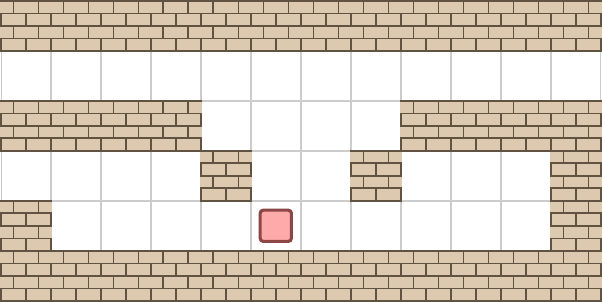}} \hspace{4em}
  \subcaptionbox{Directed Crossover}{\includegraphics[scale=0.7]{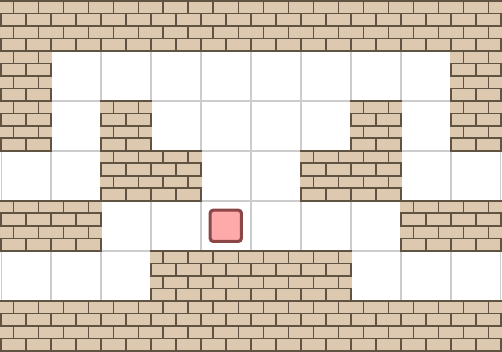}}
  \caption{Hallway connection gadgets for BoxDude.  Pushable boxes are in red. The branching hallway gadget is fully traversable from any of its three locations to the others.  The directed crossover can be traversed only from bottom-left to top-right or from bottom-right to top-left.}
  \label{fig: box hallways}
\end{figure}

\begin{figure}
	\centering
	\subcaptionbox{Directed Crossover Icon}{\includegraphics[scale=1]{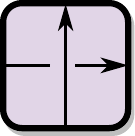}} \hfil
	\subcaptionbox{Crossover Icon}{\includegraphics[scale=1]{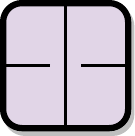}}  \hfil
	\subcaptionbox{Construction of an undirected crossover from a directed crossover.}{\includegraphics[scale=1]{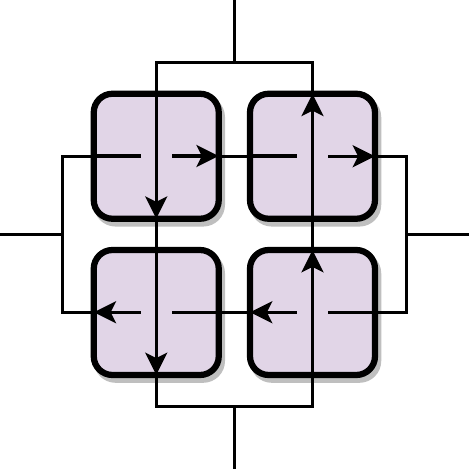}}
	\caption{Icons for directed and undirected crossovers. The undirected crossover can be constructed from four directed crossovers as shown in \cite{demainepush2F}.}
	\label{fig: directed crossover}
\end{figure}

Postselection requires us to additionally simulate the gadgets SO and MSC.
These gadgets are shown in Figure~\ref{fig: box single use doors}.

\begin{figure}
  \centering
  \subcaptionbox{SO}{\includegraphics[scale=0.8]{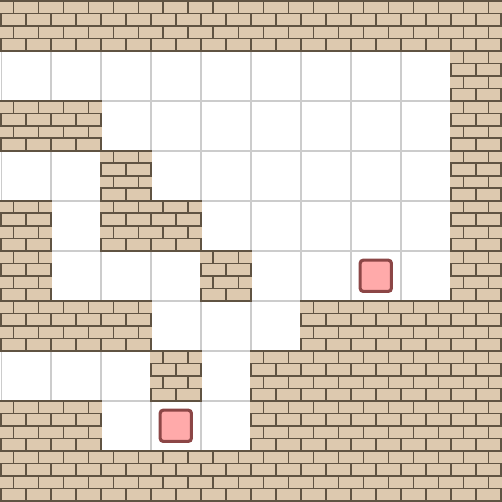}} \hspace{4em}
  \subcaptionbox{MSC}{\includegraphics[scale=0.8]{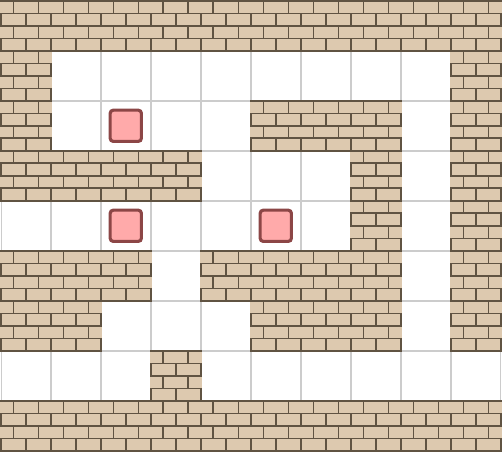}}
  \caption{SO and MSC gadgets for BoxDude.}
  \label{fig: box single use doors}
\end{figure}

Next we build a checkable \emph{leaky door} gadget.  A leaky door has two states (``open'' and ``closed''), and three locations, called ``opening'', ``entrance'', and ``exit''.  Similar to a self-closing door \cite{doors}, the gadget can be traversed in the open state from entrance to exit, but doing so transitions the door to the closed state.  In the closed state, it is not possible to enter the gadget through the entrance at all, but visiting the opening location allows the gadget to transition back to the open state.  Unlike a self-closing door, it is possible to go from the entrance to the opening location when the gadget is in the open state.  It is also always possible to go from the opening location to the exit, but doing so transitions the door to the closed state.  The full state diagram for the leaky door is shown in Figure~\ref{fig: leaky door state}.

\begin{figure}
  \centering
  \includegraphics{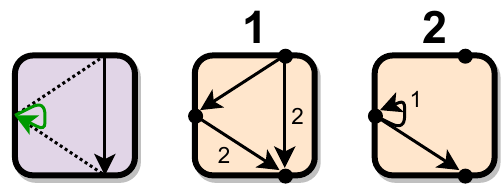}
  \caption{Icon and state diagram for the leaky door gadget.}
  \label{fig: leaky door state}
\end{figure}

The checkable leaky door is shown in Figure~\ref{fig: box leaky door}.  We apply postselection to this gadget with the checking traversal sequence
\([\text{opening} \to \text{opening},\text{entrance} \to \text{opening}]\).\footnote{The first check from \text{opening} to \text{opening} does not enforce anything but merely allows access to the location in case the gadget was last left in the closed state. The check from \text{entrance} to \text{opening} cannot be done if the gadget is in the closed state.}
We now analyze which states are \emph{broken} in the sense that this traversal sequence is impossible from those states.

\begin{itemize}
\item If the left box is further to the left than its current location, the gadget state is broken since the entrance is unusable.
\item If the left box is more than one square to the right of its current location, the gadget state is broken because the opening location is unreachable from the entrance.
\item If the two boxes are adjacent, the gadget state is broken for the same reason.
\end{itemize}

Moving the right box more than one square to the right is never advantageous for the player, so we assume it does not occur.  

We will show that the postselection of this gadget is exactly the leaky door gadget.  When the right box is in its current location, we say that the gadget is in the closed state; when it is one square to the right the gadget is in the open state.  Because the left box cannot move more than one square to the right, it follows that any traversal to the exit location must leave the gadget in the closed state.  In the closed state, no traversals are possible from the entrance without breaking the gadget by putting two boxes adjacent.  Visiting the opening allows transitioning to the open state.  In the open state, additional traversals are available from the entrance.  The agent may go from entrance to exit by using the connected opening locations to reset the gadget to the closed state and then using the right block to reach the exit.  It is also possible to leak from the entrance to the opening location, and from the opening location to the exit (transitioning to the closed state).  Thus the traversals within unbroken states are exactly those allowed by the leaky door gadget.  By Theorem~\ref{thm:postselect} the checkable leaky door, along with the SO and MSC gadgets built earlier, nonlocally simulate the leaky door.

\begin{figure}
  \centering
  \includegraphics[scale=.8]{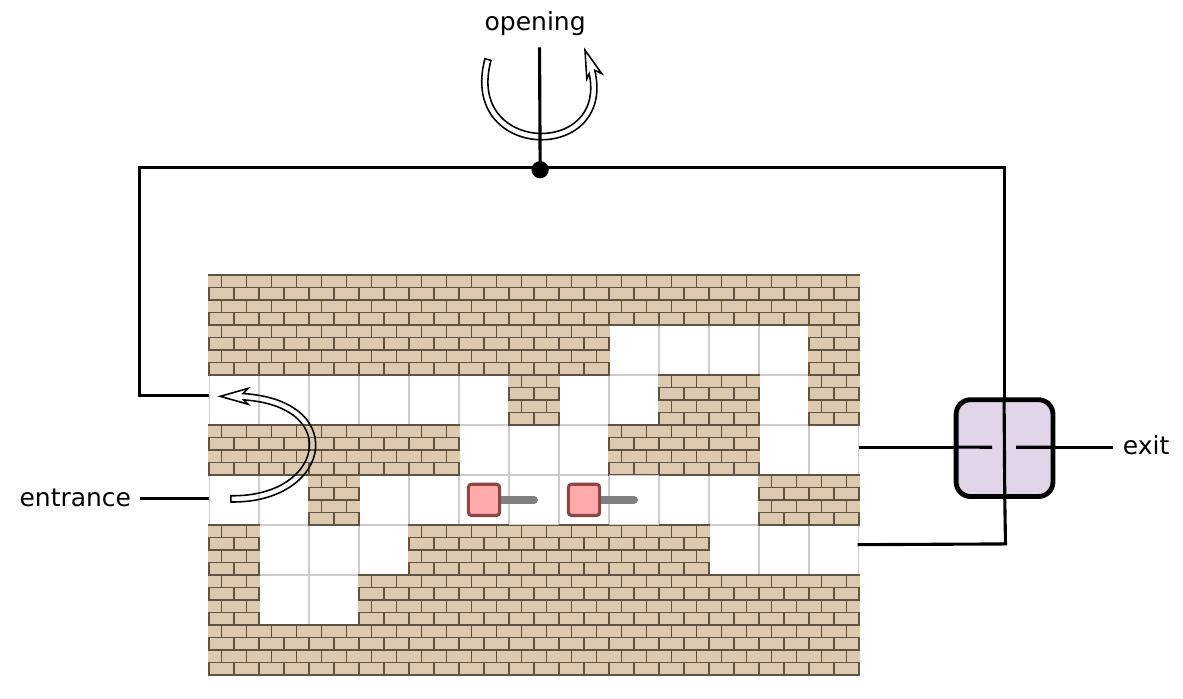}
  \caption{A checkable leaky door, shown in the closed state. The crossover and branching hallway needed to connect the top left and bottom right hallways have been abstracted.  Horizontal ``tracks'' display the range of locations for each box in unbroken states.  (The right box can move farther right but it is never advantageous to do this.)  The two boxes may not be adjacent in unbroken states.}
  \label{fig: box leaky door}
\end{figure}

We now build a 1-toggle gadget, shown in Figure~\ref{fig: 1-toggle_state}, using a pair of leaky doors. This construction is shown in Figure~\ref{fig: box 1-toggle}.  It can be seen that none of the leaks are useful to an agent traversing the gadget, since the most they accomplish is bringing the agent back to its starting location without changing any state.

\begin{figure}
  \centering
  \includegraphics{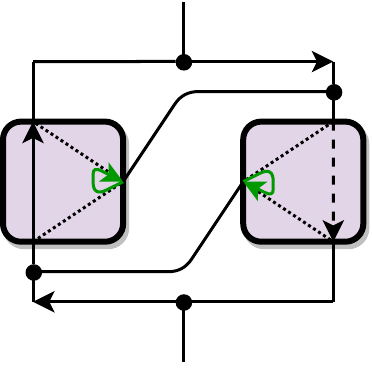}
  \caption{A 1-toggle built from leaky doors.  Solid or dashed arrows inside gadgets show the traversal from entrance to exit in an open or closed leaky door, respectively.  Green self-loops are opening locations of leaky doors.  Arrows outside gadgets are diodes.}
  \label{fig: box 1-toggle}
\end{figure}

We are now in a position to build a nondeterministic locking 2-toggle.
By Theorem~\ref{thm:nL2T_PSPACE},
reachability with this gadget is PSPACE-complete.
The final construction, shown in Figure~\ref{fig: box locking two toggle}, is quite simple in appearance; the complexity is hidden in the 1-toggles used to protect the locking 2-toggle's locations.  Traversing from A to B is only possible when the box is on the left side of the gadget, and conversely for C to D.  Since the box's position can only be changed when exiting the gadget through A or C (corresponding to which side the gadget is locked to), the gadget simulates a locking 2-toggle.  Note that this gadget cannot be broken by moving the box further to the left than its current position, since doing so renders the gadget fully untraversable.  This is because in this state location A is permanently unusable and B and D cannot be reached from inside the gadget.  The agent can only exit out of C, so that C's 1-toggle points inwards.  Since C's and D's 1-toggles always point in different directions, D is also permanently unusable.  The only remaining traversal is \(B \to C\), but this is impossible also because C's 1-toggle points inwards.

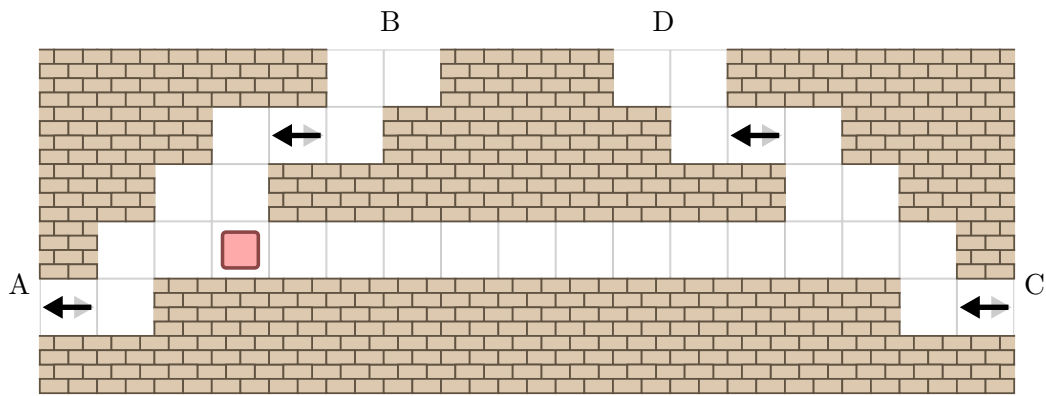
\begin{figure}
  \centering
  \vspace{1em}
  \begin{overpic}[scale=0.9]{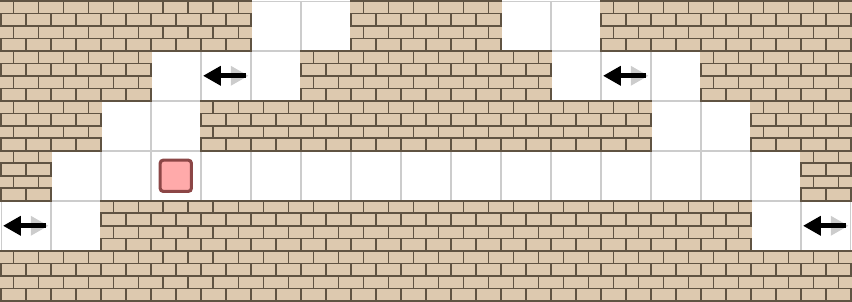}
    \put(-2,11){\makebox(0,0)[c]{\strut A}}
    \put(102,11){\makebox(0,0)[c]{\strut C}}
    \put(36,38){\makebox(0,0)[c]{\strut B}}
    \put(64,38){\makebox(0,0)[c]{\strut D}}
  \end{overpic}
  \caption{A nondeterministic locking 2-toggle, currently locked to the left side.  Locations B and C are protected with inwards-directed 1-toggles; locations A and D with outwards-directed 1-toggles.  (Note: the middle portion of the gadget would actually need to be wider than shown in this diagram in order to make enough space to route locations B and D away from each other.)}
  \label{fig: box locking two toggle}
\end{figure}

Using Theorem~\ref{thm:postselect} and Lemma~\ref{lem:compose nonlocal}, our simulations imply that the BoxDude gadgets we have explicitly built nonlocally simulate a nondeterministic locking 2-toggle. In particular, there is a polynomial-time reduction from planar reachability with nondeterministic locking 2-toggles, which is PSPACE-complete by Theorem~\ref{thm:nL2T_PSPACE}, to BoxDude. Hence BoxDude is PSPACE-complete.

\section{Push-1F is PSPACE-complete}
\label{sec:Push-1F}

In this section, we show that Push-1F is PSPACE-complete
using a reduction from planar reachability with self-closing doors
(the gadget shown in Figure~\ref{fig:self-closing door}),
which is PSPACE-complete by Theorem~\ref{thm:dir-self-closing-door-pspace}.
Recall that in this model there is no gravity, and the agent can push one block at a time in any direction.  We will make several uses of postselection from Section~\ref{sec:checkable gadget framework} in order to nonlocally simulate various gadgets along the way.

In order to use postselection, we must build single-use opening (SO) and merged single-use closing (MSC) gadgets.  We start by building a \emph{weak merged closing} gadget, based on the Lock gadget from \cite{Push100}.  The weak merged closing gadget  acts like the MSC except that the closing traversal can be performed multiple times.  We also use a gadget introduced in \cite{Push100} called a \emph{no-return} gadget.  After a no-return gadget is traversed from left to right, it cannot immediately be traversed from right to left.  However, initially traversing it from the right or traversing left to right twice breaks the gadget, making it fully traversable.  Finally, we build a \emph{weak opening} gadget.  A weak opening gadget's exit cannot be used in traversals until both of its input locations are visited separately.  Figure~\ref{fig:push base gadget states} shows the state diagrams for these gadgets, and Figure~\ref{fig:push base gadgets} shows how to implement them in Push-1F.

\begin{figure}
  \centering
  \def\scale{0.7}
  \subcaptionbox{\label{fig: weak closing state}\centering Weak merged closing}[3cm]{\includegraphics[scale=\scale]{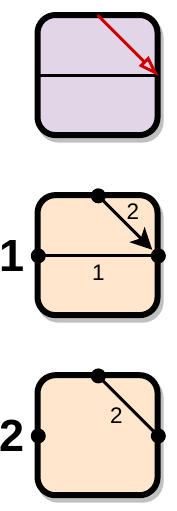}}
  \hspace{1cm}
  \subcaptionbox{\label{fig: no return state}\centering No-return}[3cm]{\includegraphics[scale=\scale]{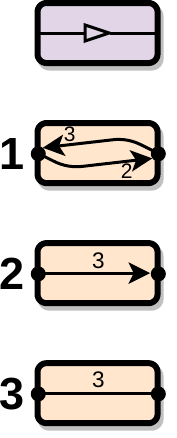}}
  \hspace{1cm}
  \subcaptionbox{\label{fig: weak opening state}\centering Weak opening}{\includegraphics[scale=\scale]{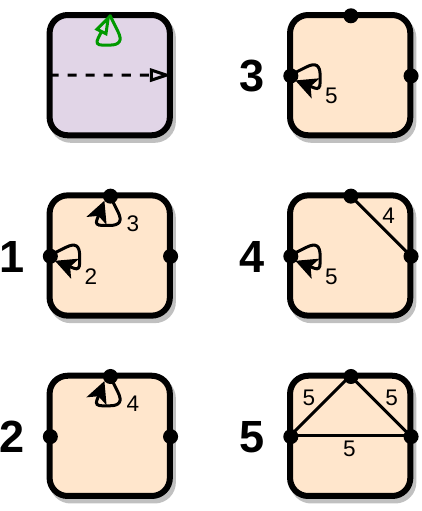}}
  \caption{Icons and state diagrams for Push-1F base gadgets.}
  \label{fig:push base gadget states}
\end{figure}

\begin{figure}
  \centering
  \def\scale{0.6}
  \subcaptionbox{\label{fig:push weak closing}\centering Weak merged closing}{\includegraphics[scale=\scale]{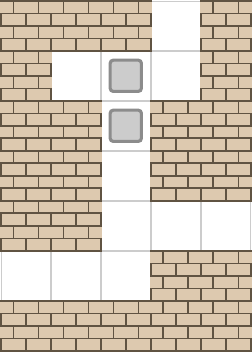}}
  \hfill
  \subcaptionbox{\label{fig:push no return}\centering No-return}{\includegraphics[scale=\scale]{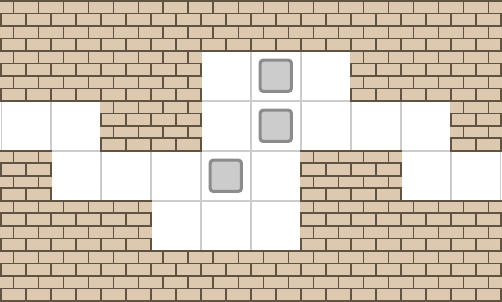}}
  \hfill
  \subcaptionbox{\label{fig:push weak opening}\centering Weak opening}{\includegraphics[scale=\scale]{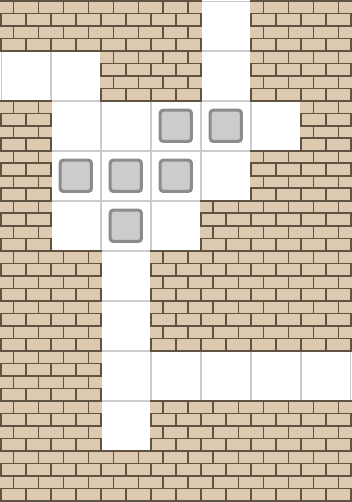}}
  \caption{Constructions of base gadgets for Push-1F.}
  \label{fig:push base gadgets}
\end{figure}

We combine the weak merged closing, no-return, and weak opening gadgets to make a dicrumbler; this allows us to simulate ordinary SO and MSC gadgets using the gadgets we have built so far.  These simulations are shown in Figure~\ref{fig:push SO and MSC}.  Having built these gadgets, we can now take advantage of the machinery of checkable gadgets.  The structure of the remaining gadget simulations used in this section is outlined in Figure~\ref{fig:push overview}.

\begin{figure}
  \centering
  \def\scale{0.6}
  \subcaptionbox{\centering Dicrumbler}{\includegraphics[scale=\scale]{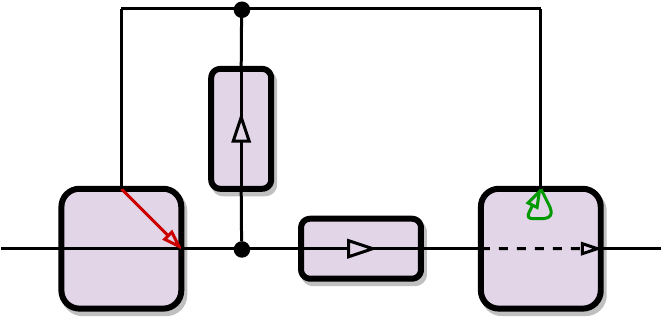}}
  \hfill
  \subcaptionbox{\centering SO}{\includegraphics[scale=\scale]{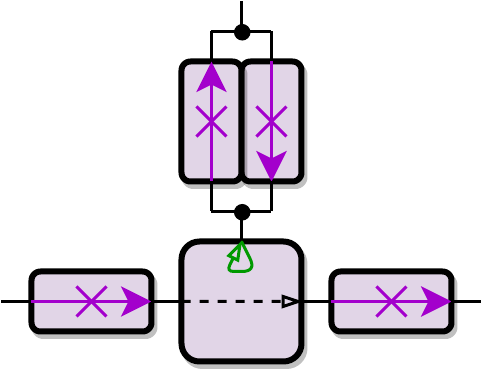}}
  \hfill
  \subcaptionbox{\centering MSC}{\includegraphics[scale=\scale]{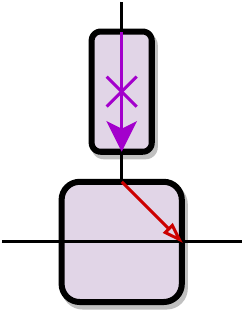}}
  \caption{Constructions of gadgets required for postselection in Push-1F.}
  \label{fig:push SO and MSC}
\end{figure}

\begin{figure}
  \centering
  \includegraphics{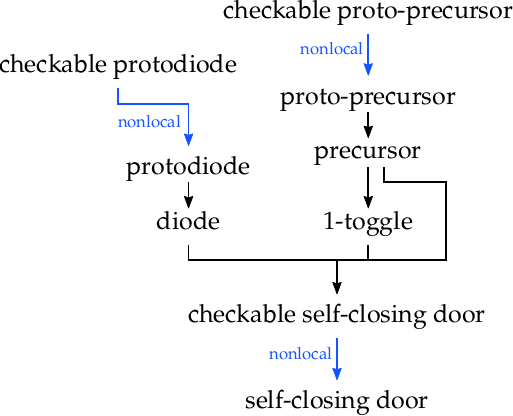}
  \caption{Overview of gadget simulations used for Push-1F.
    Black arrows show local simulations and blue arrows show nonlocal simulations.}
  \label{fig:push overview}
\end{figure}

We first nonlocally simulate a diode, which allows traversal in only one direction.  We accomplish this by building a checkable \emph{protodiode}, where the protodiode is a certain four-location gadget which easily simulates a diode.  Refer to Figure~\ref{fig:push diode}.  We apply postselection to the checkable protodiode with the checking traversals \([A \to C, D \to B]\) to nonlocally simulate the protodiode.  The nonbroken states are exactly those in which the block is confined to the middle two squares.  Connecting the bottom two locations of the protodiode yields a diode.

\begin{figure}
  \centering
    \subcaptionbox{\label{fig:push-diode-checking}\centering Checkable protodiode and checking traversals}{
  	\begin{overpic}[scale=0.8]{push-diode-check-with-lines}
  		\put(25,55){\makebox(0,0){\strut A}}
  		\put(58,55){\makebox(0,0){\strut B}}
  		\put(40, -6){\makebox(0,0){\strut C}}
  		\put(75, -6){\makebox(0,0){\strut D}}
  	\end{overpic}
  	\vspace{1em}
  }
  \hspace{2em}
  \subcaptionbox{Protodiode}{\includegraphics[scale=0.8]{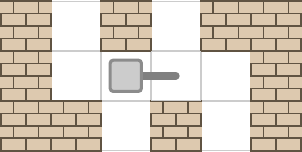}
    \vspace{1em}}
  \hspace{2em}
  \subcaptionbox{Diode}{\includegraphics[scale=0.8]{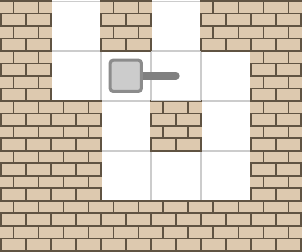}}
  \caption{Nonlocal diode simulation for Push-1F.  Horizontal tracks show where the block is allowed to move in the protodiode and diode, as if it is confined by a magical force.}
  \label{fig:push diode}
\end{figure}

We now nonlocally simulate a \emph{precursor} gadget, which will be used to build a 1-toggle and a checkable self-closing door.  The precursor's state diagram is shown in Figure~\ref{fig:push precursor state}.  We begin by building a checkable \emph{proto-precursor}, where again the proto-precursor is a certain gadget which easily simulates the precursor.  Refer to Fig~\ref{fig:push precursor}.  We apply postselection to the checkable proto-precursor with the checking traversals \([A \to D,C \to G,B \to A,B \to C]\).   We close off locations \(D\) and \(G\) during postselection by not including them in the set \(L' = \{A, B, C, E, F\}\) of locations on the proto-precursor.  The nonbroken states are exactly those in which the blocks are confined to the four center-most spaces, and the two blocks are not adjacent.  Entering a broken state is irreversible with respect to transitions on the locations in \(L'\) because \(D\) and \(G\) were excluded in \(L'\).  (If \(D\) or \(G\) were included then it would be possible to un-break the gadget from some broken states by pushing a block back into the center.)  Thus we can use postselection to nondeterministically simulate the proto-precursor; joining its upper three locations together yields the precursor gadget.  Additionally, closing the top location of the precursor gadget produces a 1-toggle.

\begin{figure}
  \centering
  \bigskip
  \subcaptionbox{\label{fig:push-precursor-checking}Checkable proto-precursor and checking traversals.  Locations excluded from \(L'\) are marked with an X.}{
	\begin{overpic}[scale=0.8]{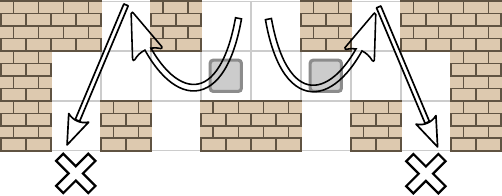}
	  \put(25,42){\makebox(0,0){\strut A}}
	  \put(50,42){\makebox(0,0){\strut B}}
	  \put(75,42){\makebox(0,0){\strut C}}
	  \put(15,-4){\makebox(0,0){\strut D}}
	  \put(35,-4){\makebox(0,0){\strut E}}
	  \put(65,-4){\makebox(0,0){\strut F}}
	  \put(85,-4){\makebox(0,0){\strut G}}
	\end{overpic}
	\vspace{1em}
  }
  \hspace{3em}
  \subcaptionbox{Proto-precursor}{\includegraphics[scale=0.8]{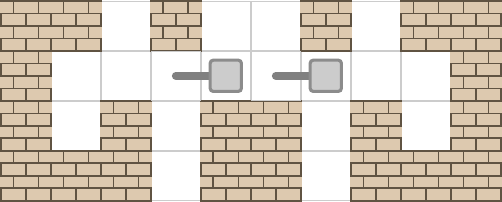}}
  \\
  \vspace{1em}
  \subcaptionbox{Precursor}{\includegraphics[scale=0.8]{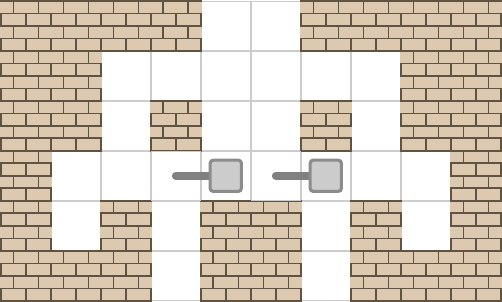}}
  \hspace{3em}
  \subcaptionbox{\label{fig:push precursor state}Icon and state diagram for precursor}{\includegraphics[scale=0.8]{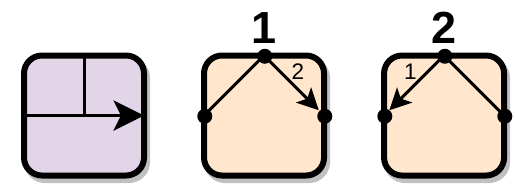}}
  \caption{Nonlocal precursor simulation for Push-1F.  As before, horizontal tracks in the proto-precursor and precursor show spaces to which blocks are magically confined.  The magical force also prevents the pair of blocks in the proto-precursor and precursor from being adjacent.}
  \label{fig:push precursor}
\end{figure}
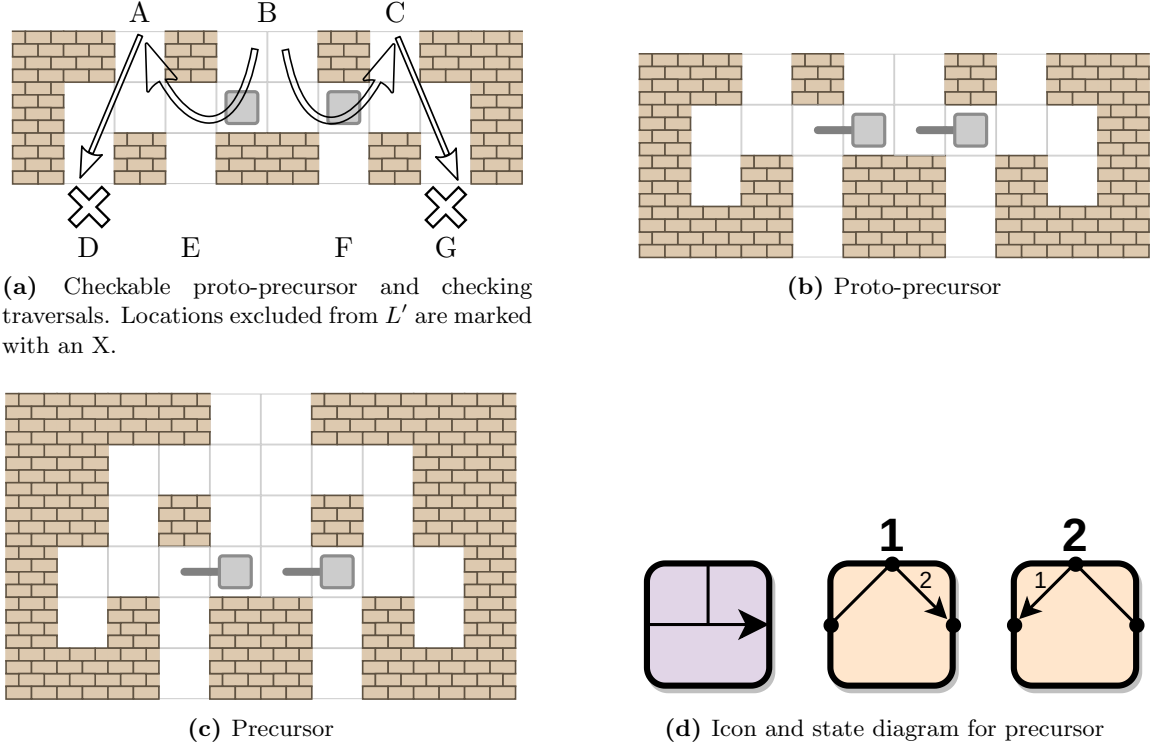

Finally, we nonlocally simulate a self-closing door.  Our construction of a checkable self-closing door is shown in Figure~\ref{fig:push scd}.  This gadget is almost identical to a self-closing door, except that it permits a traversal from the opening location to the exit location exactly once, after which the gadget is fully untraversable.  We eliminate this problem by applying postselection with the checking traversal sequence \([\text{opening} \to \text{opening},\text{entrance} \to \text{exit}]\).  The sole broken state is the fully untraversable one arising from the aforementioned undesired traversal.  If we imagine that a magical force prevents the gadget from being left in such a state, then we obtain exactly a self-closing door.

\begin{figure}
	\centering
	\includegraphics{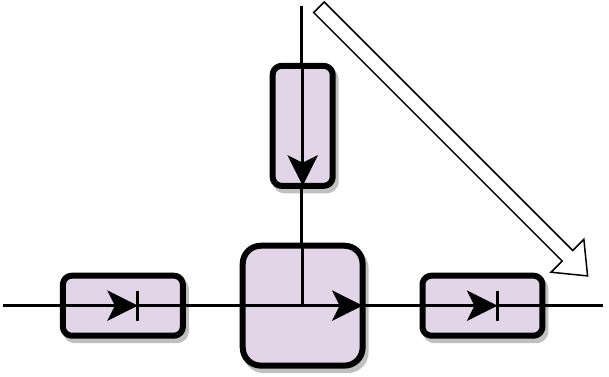}
	\caption{Checkable self-closing door for Push-1F using the precursor gadget, two diodes, and a 1-toggle.}
	\label{fig:push scd}
\end{figure}

We have demonstrated a series of planar, nonlocal gadget simulations culminating in the planar nonlocal simulation of a self-closing door.
Because planar reachability through systems of self-closing doors
is PSPACE-complete by Theorem~\ref{thm:dir-self-closing-door-pspace},
so is Push-1F.

\section{Push-$k$ for $k \ge 2$}
\label{sec:Push-k}

In this section, we show that Push-$k$ is PSPACE-complete for any $k \ge 2$,
by following the reduction to Push-1F from Section~\ref{sec:Push-1F}
but with different gadgets.
In Push-$k$, all blocks are movable in principle,
but in fact any block contained in a $(k+1) \times (k+1)$ square of blocks
can never move
(generalizing a similar observation made for Push-1 in \cite{Push100}).
In our gadget figures, we draw such blocks as fixed (bricks),
even though there is no actual distinction in the puzzle.

First we present analogs of the base gadgets from
Figure~\ref{fig:push base gadgets} for Push-$k$ for any $k \ge 2$.
Figure~\ref{fig:push-k base gadgets} shows the new constructions,
illustrated for Push-2 and Push-3 so that the pattern for general $k$ is clear.
The weak merged closing gadget is the natural generalization
of Figure~\ref{fig:push weak closing} and works the same,
while the no-return gadget is simplified relative to Figure~\ref{fig:push no return}.
(In fact, we do not need the no-return gadget for Push-$k$ with $k \geq 3$,
because we will shortly construct a diode for these puzzles.)
The weak opening permits a slightly different set of traversals than those in Figure~\ref{fig: weak opening state}, but it works for our purposes.
Out of these base gadgets, we build the SO and MSC gadgets
needed by the checkable gadgets framework
exactly as we did for Push-1 in Figure~\ref{fig:push SO and MSC}.
Thus we can use postselection.

\begin{figure}
  \centering
  \def\scale{0.5}
  \subcaptionbox{\label{fig:push-2 weak closing}\centering Weak merged closing}{\includegraphics[scale=\scale]{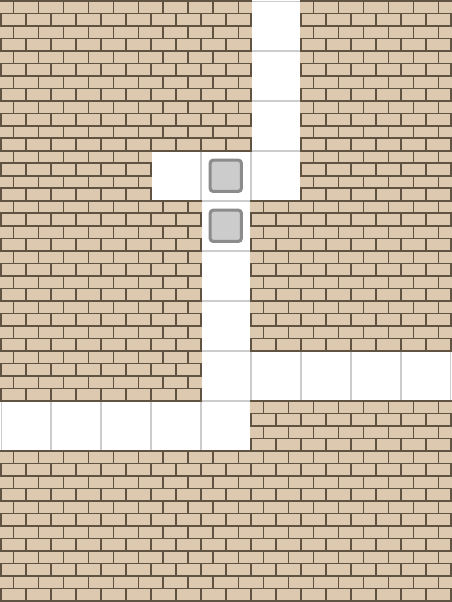}}  \hfil
  \subcaptionbox{\label{fig:push-2 no return}\centering No-return}{\includegraphics[scale=\scale]{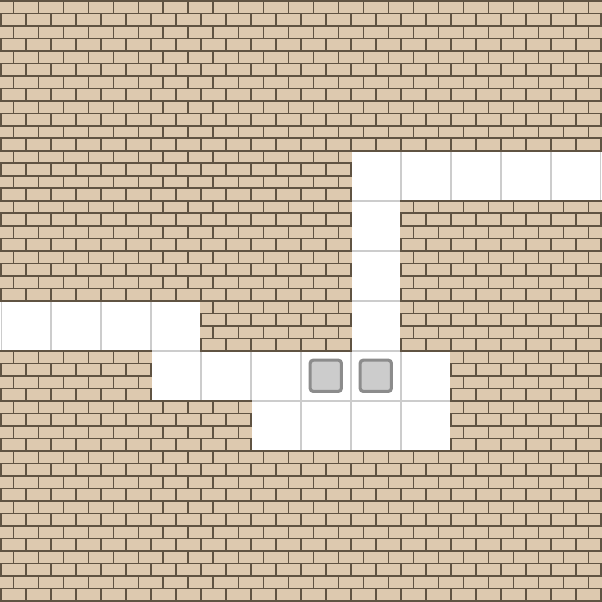}}  \hfil
  \subcaptionbox{\label{fig:push-2 weak opening}\centering Weak opening}{\includegraphics[scale=\scale]{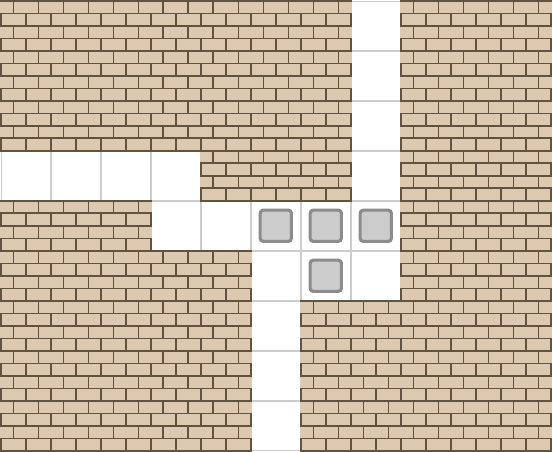}}

  \medskip
  \def\scale{0.45}
  \subcaptionbox{\label{fig:push-3 weak closing}\centering Weak merged closing}{\includegraphics[scale=\scale]{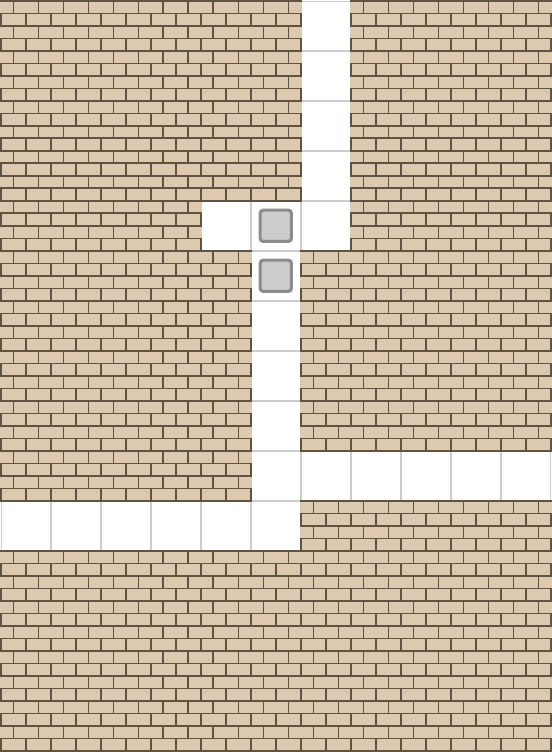}}  \hfil
  \subcaptionbox{\label{fig:push-3 no return}\centering No-return}{\includegraphics[scale=\scale]{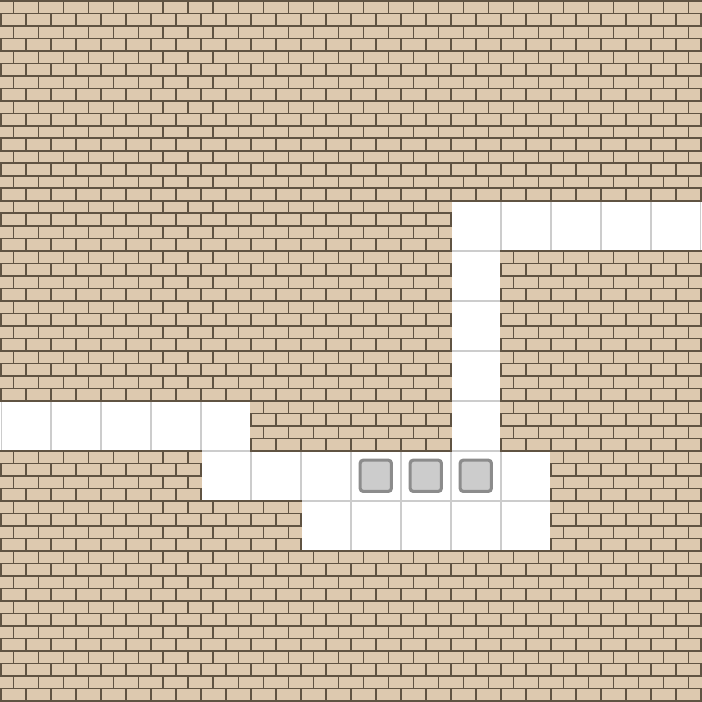}}  \hfil
  \subcaptionbox{\label{fig:push-3 weak opening}\centering Weak opening}{\includegraphics[scale=\scale]{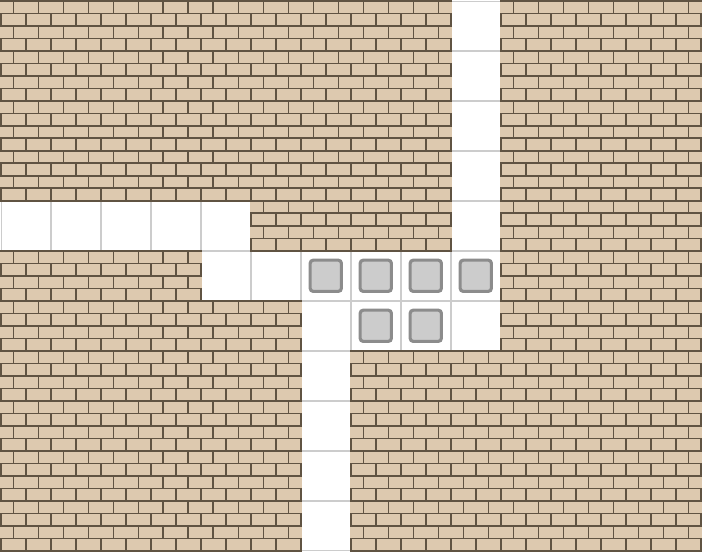}}
  \caption{Constructions of base gadgets for Push-$k$ for any $k \ge 2$,
    illustrated for Push-2 (a--c in top row) and Push-3 (d--f in bottom row).}
  \label{fig:push-k base gadgets}
\end{figure}

Next we present the two checkable gadgets that we used to build a
self-closing door: a diode and a precursor.
Figure~\ref{fig:push-k diode} shows a checkable proto-diode gadget
for Push-2 and a diode gadget for Push-3, the latter of which generalizes to arbitary $k \ge 3$.
Figure~\ref{fig:push-k precursor} shows checkable proto-precursor gadgets
for Push-2 and Push-3, the latter of which generalizes to arbitrary $k \ge 3$.
The Push-2 version uses the same checking sequence as the corresponding Push-1F gadget, whereas the version for $k \ge 3$ uses the simplified checking sequence $[B \to B, A \to B]$ which is sufficient to ensure correct behavior.

\begin{figure}
  \centering
  \subcaptionbox{Push-2 checkable proto-diode, which requires the same checking as Figure~\ref{fig:push-diode-checking}.}{  \begin{overpic}[scale=0.6]{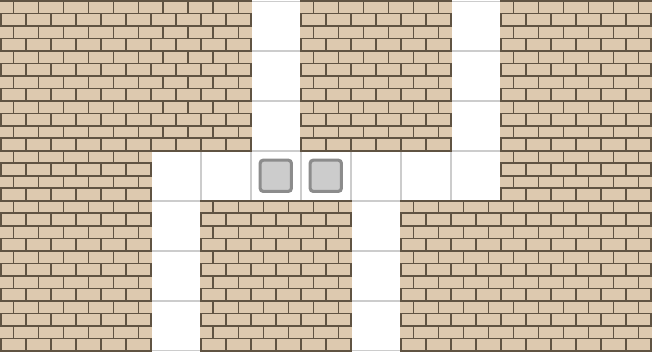}
    \put(42.5,57){\makebox(0,0){\strut A}}
    \put(73.0,57){\makebox(0,0){\strut B}}
    \put(27  ,-4){\makebox(0,0){\strut C}}
    \put(57.8,-4){\makebox(0,0){\strut D}}
  \end{overpic}
  \vspace{1em}}  \hfill
  \subcaptionbox{Push-3 diode, which generalizes to Push-$k$ for any $k \geq 3$ by duplicating the row and column marked $\star$.}{  \begin{overpic}[scale=0.6]{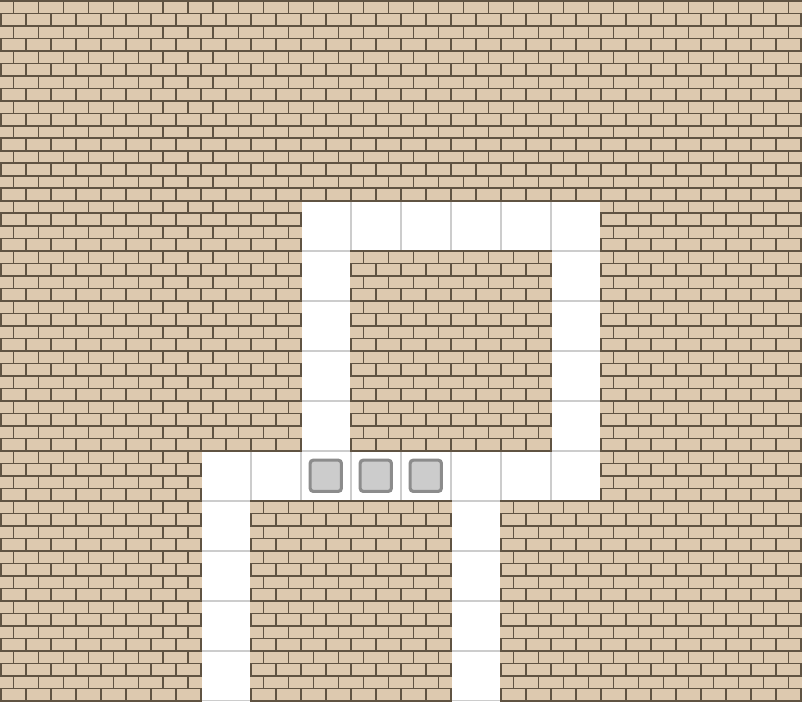}
	\put(-3,47){\makebox(0,0){\strut $\star$}}
	\put(53.25,-3){\makebox(0,0){\strut $\star$}}
    \end{overpic}
  \vspace{1em}}
  \caption{Diode gadget for Push-$k$ for any $k \ge 2$.}
  \label{fig:push-k diode}
\end{figure}

\begin{figure}
  \centering
  \vspace{1em}
  \subcaptionbox{Push-2 checkable proto-precursor, which requires the same checking as Figure~\ref{fig:push-precursor-checking}.}{  \begin{overpic}[scale=0.45]{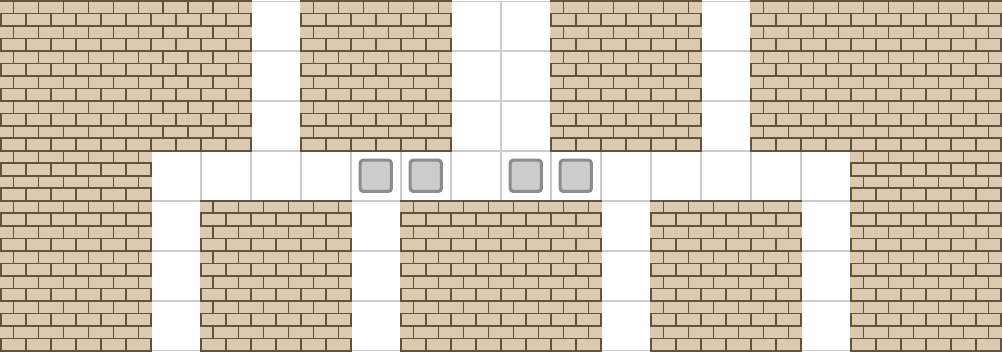}
	  \put(27.5,38){\makebox(0,0){\strut A}}
	  \put(50  ,38){\makebox(0,0){\strut B}}
	  \put(72.5,38){\makebox(0,0){\strut C}}
	  \put(17.5,-4){\makebox(0,0){\strut D}}
	  \put(37.5,-4){\makebox(0,0){\strut E}}
	  \put(62.5,-4){\makebox(0,0){\strut F}}
	  \put(82.5,-4){\makebox(0,0){\strut G}}
  \end{overpic}
  \vspace{1em}}  \hfill
  \subcaptionbox{Push-3 checkable proto-precursor, with checking traversals $[B \to B, A \to B]$.  Generalizes to Push-$k$ for any $k \geq 3$ by duplicating the columns marked $\star$.}{  \begin{overpic}[scale=0.45]{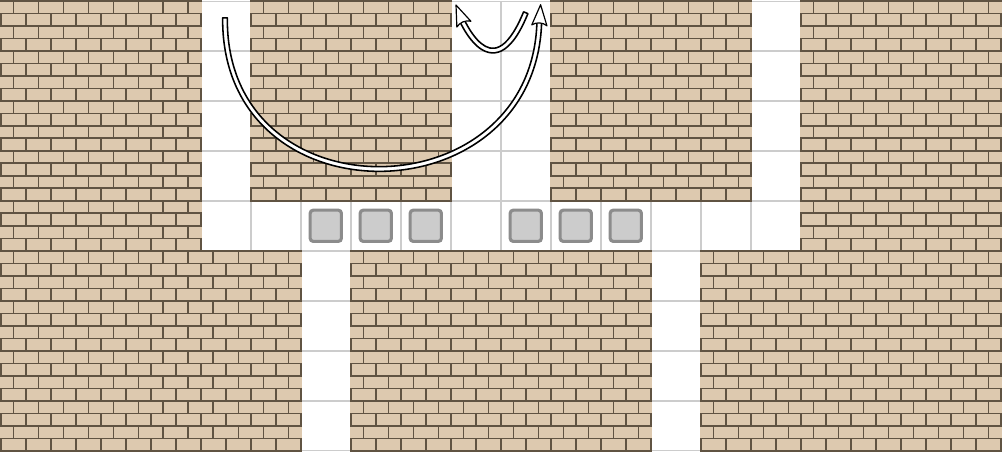}
	  \put(22.5,48){\makebox(0,0){\strut A}}
	  \put(50  ,48){\makebox(0,0){\strut B}}
	  \put(77.5,48){\makebox(0,0){\strut C}}
	  \put(32.5,-4){\makebox(0,0){\strut E}}
	  \put(67.5,-4){\makebox(0,0){\strut F}}
	  \put(42.5,-3){\makebox(0,0){\strut $\star$}}
	  \put(57.5,-3){\makebox(0,0){\strut $\star$}}
  \end{overpic}
  \vspace{1em}}
  \caption{Precursor gadget for Push-$k$ for any $k \ge 2$.}
  \label{fig:push-k precursor}
\end{figure}
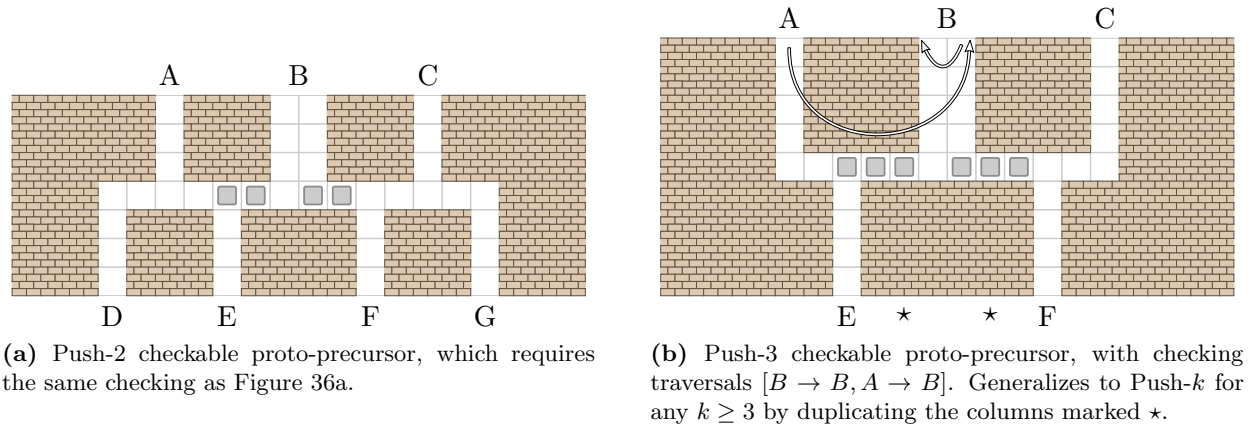

Using these gadgets, we can build a self-closing door exactly as we did for Push-1F in Figure~\ref{fig:push scd}.
Therefore Push-$k$ is PSPACE-complete for any $k \ge 2$.

\section{Open Problems}
\label{sec:Open Problems}

\definecolor{header}{rgb}{0.29,0,0.51}
\definecolor{gray}{rgb}{0.85,0.85,0.85}
\definecolor{hard}{rgb}{1,0.85,0.85}
\definecolor{open}{rgb}{1,1,0.85}
\definecolor{easy}{rgb}{0.85,0.85,1}
\def\header#1{\multicolumn{1}{c}{\cellcolor{header}\textcolor{white}{\textbf{#1}}}}
\def\OPEN{\cellcolor{open}\em OPEN}
\def\HARD{\cellcolor{hard}}

Among the Push-$k$(F) family puzzles, as summarized in Table~\ref{tab:push},
just two cases remain open: Push-1 and Push-$*$ without fixed blocks.
As in our Push-$k$ proofs of Section~\ref{sec:Push-k},
Push-1 can easily simulate fixed blocks using $2 \times 2$ arrangements of
movable blocks, so we ``only'' need to make all fixed areas two blocks thick.
Our constructions of the gadgets SO and MSC needed to apply postselection
all use two-block thick spacing, so we have shown that postselection is
available for Push-1 gadgets.
Unfortunately, our postselected constructions for Push-1F critically use one-block-thick spacing.
With Push-$*$, it seems difficult to even build gadgets that last for more
than one use.
Both Push-1 and Push-$*$ are known to be NP-hard
\cite{Push100,hoffmann-2000-pushstar,demaine2003pushing}.
We conjecture that Push-1 is PSPACE-complete and Push-$*$ is NP-complete.

\begin{table}
  \centering
  \small
    \begin{tabular}{|l|c|c|c|}
    \multicolumn{1}{l|}{}
    & \header{Push-1} & \header{Push-$k$ for $k \geq 2$} & \header{Push-$*$} \\
    \hline
    \header{No fixed blocks} & \OPEN & \HARD PSPACE-complete [\S\ref{sec:Push-k}] & \OPEN \\
    \hline
    \header{Fixed blocks (F)} & \HARD PSPACE-complete [\S\ref{sec:Push-1F}] & \HARD PSPACE-complete [\S\ref{sec:Push-k}] & \HARD PSPACE-complete \cite[\S\ref{sec:Push-*F}]{bremner1994motion} \\
    \hline
  \end{tabular}
  \caption{Summary of known complexity results for Push-$k$(F) puzzles.}
  \label{tab:push}
\end{table}

In the $\cdots$Dude puzzle family,
the main open problem is the related block storage question, named $\cdots$Duderino in \cite{barr2021block}, in which the blocks have target locations to occupy. This is comparable to the difference between Push-1F and Sokoban. It is generally expected that the storage version of block-pushing puzzles is at least as hard as reaching a single goal location; however, this result does not directly follow. We believe using the reconfiguration version of the gadgets framework from \cite{DBLP:conf/walcom/AniDDHL22} may help build a gadget-based proof.

Another open question relates to the technique of postselected gadgets. When defining a postselected gadget, we only specified a single traversal sequence to be checked. It seems likely that one could enforce the choice of one of several possible sequences using more complex constructions like those found in the SAT reduction for DAG gadgets in \cite{demaine2018general}. Are there cases where this sort of flexibility is useful?

\section*{Acknowledgments}

This work was initiated during extended problem solving sessions
with the participants of the MIT class on
Algorithmic Lower Bounds: Fun with Hardness Proofs (6.892)
taught by Erik Demaine in Spring 2019.
We thank the other participants for their insights and contributions.

We would like to thank our reviewers for their detailed and useful feedback.

We would like to thank Aaron Williams for useful discussion including how to restructure the paper and how to better present the results and checkable gadget framework.

We would like to thank the authors of \cite{bremner1994motion} for permission
to include versions of their gadgets in Section~\ref{sec:Push-*F}.

Figures produced using SVG Tiler (\url{https://github.com/edemaine/svgtiler}),
diagrams.net, and Inkscape.

\bibliographystyle{alpha}
\bibliography{thebib}

\appendix
\magicappendix

\end{document}